\documentclass[
								12pt,
								reprint,
								aip,
								jmp,
								onecolumn,
								a4paper
							]{revtex4-1}

\usepackage{lmodern,enumerate}
\usepackage[fleqn]{amsmath}
\usepackage{amssymb,amsthm,amsfonts,mathrsfs}

\usepackage{graphicx,epstopdf,epsfig}

\usepackage[hidelinks]{hyperref}

\newcommand{\nat}{\mathbb{N}}
\newcommand{\real}{\mathbb{R}}
\newcommand{\complex}{\mathbb{C}}
\newcommand{\integer}{\mathbb{Z}}

\newcommand{\hilbert}{\mathscr{H}}

\newcommand{\cstar}{$\mathrm{C}^{\star}$-algebra }
\newcommand{\czero}{$\mathcal{C}_{0}$-semigroup }
\newcommand{\bh}{\mathscr{B}(\hilbert)}
\newcommand{\bs}{\mathscr{X}}
\newcommand{\bx}{\mathscr{B}(\bs)}
\newcommand{\kh}{\mathscr{K}(\hilbert)}
\newcommand{\traceclass}{\mathscr{B}_{1}(\hilbert)}

\newcommand{\enmor}[1]{\mathrm{Lin} (#1)}
\newcommand{\kernel}[1]{\mathrm{Ker} (#1)}
\newcommand{\range}[1]{\mathrm{Ran} (#1)}
\newcommand{\domain}[1]{\mathrm{Dom} (#1)}
\newcommand{\dimension}[1]{\mathrm{dim} \, #1}
\newcommand{\iprod}[2]{\langle #1, #2 \rangle}

\newcommand{\tr}[1]{\mathrm{tr}\left\{ #1 \right\}}
\newcommand{\ptr}[2]{\mathrm{Tr}_{#1}\left\{#2\right\}}
\newcommand{\id}[1]{\mathbb{\mathrm{I}}_{#1}}
\newcommand{\comm}[2]{\left[#1,#2\right]}
\newcommand{\acomm}[2]{\left\{#1,#2\right\}}
\newcommand{\slim}[1]{s - \lim_{#1}}

\newcommand{\hc}{\mathrm{h.c.}}
\newcommand{\adj}{\star}

\newcommand{\dual}{\prime}

\newcommand{\hav}{\bar{H}}

\newcommand{\schr}{Schr\"{o}dinger }

\newcommand{\dm}[1]{\rho_{#1}}
\newcommand{\dmr}{\omega_{\mathrm{R}}}
\newcommand{\lgen}[2]{\mathcal{L}_{#1}(#2)}
\newcommand{\algen}[2]{\mathcal{L}_{#1}^{\adj} (#2)}
\newcommand{\lgens}[1]{\mathcal{L}_{#1}}
\newcommand{\algens}[1]{\mathcal{L}_{#1}^{\adj}}
\newcommand{\hsup}[2]{-i \comm{H_{\mathrm{S}}(#1)}{#2}}
\newcommand{\usup}[3]{\mathcal{U}_{#1, #2}(#3)}
\newcommand{\usups}[2]{\mathcal{U}_{#1, #2}}
\newcommand{\usupgen}[4]{\mathcal{U}^{#1}_{#2,#3}(#4)}
\newcommand{\ausupgen}[4]{(\mathcal{U}^{#1}_{#2,#3})^{-1}(#4)}
\newcommand{\usupsgen}[3]{\mathcal{U}^{#1}_{#2,#3}}
\newcommand{\ausupsgen}[3]{(\mathcal{U}^{#1}_{#2,#3})^{-1}}
\newcommand{\ausup}[3]{\mathcal{U}_{#1, #2}^{-1}(#3)}
\newcommand{\ausups}[2]{\mathcal{U}_{#1, #2}^{-1}}
\newcommand{\diss}[2]{\mathcal{D}_{#1}(#2)}

\newcommand{\disss}[1]{\mathcal{D}_{#1}}

\newcommand{\qdm}[3]{\Lambda_{#1, #2}(#3)}
\newcommand{\aqdm}[3]{\Lambda_{#1, #2}^{\adj}(#3)}
\newcommand{\qdms}[2]{\Lambda_{#1, #2}}
\newcommand{\aqdms}[2]{\Lambda_{#1, #2}^{\adj}}
\newcommand{\hilberts}{\hilbert_{\mathrm{S}}}
\newcommand{\hilbertr}{\hilbert_{\mathrm{R}}}
\newcommand{\hilbertsr}{\hilberts \otimes \hilbertr}
\newcommand{\hams}{H_{\mathrm{S}}}
\newcommand{\hamr}{H_{\mathrm{R}}}
\newcommand{\hami}{H_{\mathrm{int.}}}
\newcommand{\tcsystem}{\mathscr{B}_{1}(\hilberts)}
\newcommand{\tcenv}{\mathscr{B}_{1}(\hilbertr)}
\newcommand{\tcjoint}{\mathscr{B}_{1}(\hilbertsr)}

\newcommand{\iplgen}[1]{\tilde{\mathcal{L}}(#1)}
\newcommand{\aiplgen}[1]{\tilde{\mathcal{L}}^{\adj}(#1)}
\newcommand{\iplgens}{\tilde{\mathcal{L}}}
\newcommand{\aiplgens}{\tilde{\mathcal{L}}^{\adj}}

\newcommand{\ipdissprime}[1]{\tilde{\mathcal{D}}^{\prime}(#1)}

\newcommand{\ipqdm}[3]{\tilde{\Lambda}_{#1,#2}(#3)}

\newcommand{\ipqdms}[2]{\tilde{\Lambda}_{#1,#2}}

\newcommand{\ipdm}[1]{\tilde{\rho}_{#1}}
\newcommand{\ips}[2]{\tilde{S}_{#1}(#2)}
\newcommand{\ipr}[2]{\tilde{R}_{#1}(#2)}
\newcommand{\iphami}[1]{\tilde{H}_{\mathrm{int.}}(#1)}

\begin{document}

\theoremstyle{plain}
\newtheorem{theorem}{Theorem}
\newtheorem{lemma}{Lemma}
\newtheorem{proposition}{Proposition}

\theoremstyle{definition}
\newtheorem{definition}{Definition}
\newtheorem{corollary}{Corollary}
\newtheorem{fact}{Fact}
\newtheorem{example}{Example}

\theoremstyle{remark}
\newtheorem{remark}{Remark}

\title{On the application of Floquet theorem in development of time-dependent Lindbladians}

\author{Krzysztof Szczygielski}
\email[e-mail: ]{fizksz@ug.edu.pl}
\homepage[homepage: ]{www.fizksz.strony.ug.edu.pl}

\affiliation{Institute of Theoretical Physics and Astrophysics, University of Gdansk, Wita Stwosza 57, 80-952 Gdansk, Poland}

\date{\today}

\begin{abstract}

In this paper, the mathematical framework providing a description of completely positive and trace preserving dynamics of open quantum systems is addressed. Special case of time-dependent Lindbladian governed by periodic Hamiltonian is concerned. It is proven, that appropriate trace preserving dynamical map of such periodically driven system may be constructed by application of Floquet theory. Appropriate Lindbladian and Markovian master equation in weak coupling regime and resulting dynamical map are constructed. Some examples of application of developed technique are given.

\end{abstract}

\pacs{02.30.Hq, 02.30.Px, 02.30.Sa, 02.30.Tb, 03.65.Aa, 03.65.Db, 03.65.Yz}
\keywords{Floquet theorem, open quantum systems, completely positive maps, quantum dynamical semigroups, time-dependent generators}

\maketitle

\section{Introduction}
\label{Introduction}

A traditional mathematical apparatus used in a theory of open quantum systems was derived in 1970's and 1980's by Lindblad\cite{Lindblad76}, Kossakowski, Sudarshan and Gorini\cite{GKS76}, with further results provided by Davies\cite{Davies74,Davies1976} and others. In this framework, the evolution of density operator $\dm{t}$ of some physical system is considered in terms of contraction semigroup (with unity) of completely positive and trace preserving maps $\qdms{t}{t_0}$, with dense domain, acting on Banach space $\traceclass$ of trace-class operators in a way such that $\dm{t} = \qdm{t}{t_0}{\dm{t_0}}$. If a system's Hamiltonian is constant in time, as was originally considered, the structure of such semigroup is known to be of a form $\qdms{t}{t_0} = e^{(t-t_0)\mathcal{L}}$ with $\mathcal{L}$ known as Lindblad-Gorini-Kossakowski-Sudarshan generator\cite{Lindblad76,GKS76}. However, in general case of time-dependent Hamiltonian $H(t)$, obtaining an analytically rigorous expression for $\qdms{t}{t_0}$ may be even impossible and only formal solution can be proposed.
In this paper we focus on a much smaller class of systems with \emph{periodic} Hamiltonians and show that in such case one can actually provide a mathematically exact expression for $\qdms{t}{t_0}$. Its particular form was proposed \emph{ad hoc} in 2006 by Alicki, Lidar and Zanardi (see ref. \onlinecite{AlickiLidarZanardi2006}), however lacking a proper derivation. Here we derive $\lgens{t}$ from microscopic model on a background of \emph{Floquet theorem}, applicable to systems of differential equations with periodic coefficients. As a result, we show that this proposed intuitive approach is indeed mathematically strict and self-consistent.
\par
The article is structured as follows: in section \ref{sec:FloquetTheory} a general theory of ordinary differential equations is presented; we provide some basic notions and concepts. Next, the Floquet formalism is introduced, both in case of general, abstract setting of equations in Banach space, and in special case of \schr equation in Hilbert space (in subsect. \ref{subsec:QMperiodicH}). In section \ref{sec:OQS}, a short introduction to open quantum systems is given. The main result comes in subsection \ref{subsec:MainResult}, where a microscopic derivation of Markovian master equation in weak coupling regime under periodic Hamiltonian is presented (some technical details are available in appendix). Some various properties of obtained Lindbladian and dynamical map are elaborated. Finally, in section \ref{sec:Applications} some examples are given.
\par
With letters $\bs$ and $\hilbert$ we denote general Banach and Hilbert spaces, respectively. $\mathscr{B}(X)$ traditionally denotes an algebra of all bounded linear operators over space $X$, while $\enmor{X}$ will denote a space off all linear maps over $X$. When $\dimension{X} = k < \infty$, $\enmor{X} = \mathscr{B}(X) = \mathscr{M}_{k}(\mathbb{F})$, an algebra of all $k$-by-$k$ matrices with coefficients from field $\mathbb{F}$ ($\complex$ in general). $\id{X}$ stands for identity operator over $X$. Banach space of all trace-class operators over $\hilbert$ is marked with $\traceclass$. For given operator $A$, the domain, kernel, range and hermitian adjoint of $A$ are respectively denoted as $\domain{A}$, $\kernel{A}$, $\range{A}$ and $A^{\adj}$. With regular, uppercase letters ($A$, $B$, ...) we denote linear operators over Hilbert space. Calligraphic font ($\mathcal{A}$, $\mathcal{B}$, ...) denotes linear maps over $\traceclass$ and $\bh$.

\section{Ordinary differential equations}
\label{sec:FloquetTheory}

\subsection{Basic theorems and notation}

Here we present some framework of theory of ordinary differential equations. A deeper insight into this domain can be easily found in numerous literature\cite{Krein72,Chicone99}, therefore we omit some of technical details. We consider a general problem of homogeneous, ordinary differential equation (ODE) with periodic coefficient in some Banach space $\bs$,
\begin{equation}\label{ODEs}
	\frac{d}{dt}f(t) = A_{t}f(t), \qquad f(t_0) = f_{0} \in \bs,
\end{equation}
where the mapping $t \longmapsto A_{t} \in \bx$ is at least piecewise-continuous and periodic with period $T$, $A_{t+T} = A_{t}$ for all $t$. Let $\domain{A_{t}}$ be dense in $\bs$ and $f(t) \in \domain{A_{t}}$.
\vskip\baselineskip
By \emph{fundamental solution} $\Phi$ of such ODE \cite{Krein72,Chicone99} we will understand a continuously differentiable map $\real \ni t \longrightarrow \Phi_{t} \in \enmor{\bs}$, which satisfies\cite{Tanabe59,Tanabe60} the corresponding ODE, $\frac{d}{dt}\Phi_{t} = A_{t} \Phi_{t}$, and which is nonsingular, i.e. $\Phi_{t}^{-1}$ exists for all $t$. If, additionally, $\Phi_{t_0} = \id{\bs}$ for some $t_0 \in \real$ then $\Phi$ will be called a \emph{principal fundamental solution}. When $\dimension{\bs}$ is finite this condition implies that it has non-zero \emph{Wro\'{n}skian} $W_{\Phi}(t) = \det{\Phi_{t}}$ and $\Phi$ is commonly called a fundamental matrix. It can be proven that if $\Phi_{t}$ and $\Phi^{\prime}_{t}$ are both fundamental solutions of ODE \eqref{ODEs}, then there exists a constant operator $C$ such that $\Phi^{\prime}_{t} = \Phi_{t} C$. By virtue of that statement one can introduce a somehow more convenient object, a continuously differentiable mapping $(t,t_0) \longmapsto U_{t,t_0} = \Phi_{t} \Phi_{t_0}^{-1}$ known as \emph{state transition operator} which itself is a particular solution of the ODE,
\begin{equation}\label{eq:ODEstateTransitionOperator}
	\frac{d}{dt} U_{t,t_0} = A_{t} U_{t,t_0}, \qquad U_{t_0,t_0} = \id{\bs}
\end{equation}
so it is also a fundamental solution and, since it satisfies the boundary condition $U_{t_0, t_0} = \id{\bs}$, even principal. It satisfies \emph{Chapman -- Kolmogorov properties}\cite{Krein72,Chicone99}, namely $U_{t,t_0}^{-1} = U_{t_0, t}$, $U_{t,t_0} = U_{t,s}U_{s,t_0}$ for $t_0 \leqslant s \leqslant t$. It is known, that if $f(t)$ represents any general solution of \eqref{ODEs}, then there exist $\phi \in \bs$ and $t_0 \in \real$ such that $f(t) = \Phi_{t} \phi$ and $f(t) = U_{t,t_0} f(t_0)$ are particular solutions, and by proper choice of $t_0$ one can always get $\phi = f(t_0)$. Therefore the theory which is going to be outlined later, will be formulated entirely through $U_{t,t_0}$.
\vskip\baselineskip
\noindent For our purpose it is convenient to express a solution $U_{t,t_0}$ in form of ordered exponential\cite{Tanabe60,HuelgaRivas2012},
\begin{equation}\label{eq:OrderedExponential}
	U_{t,t_0} = \mathcal{T} \exp{\int\limits_{t_0}^{t} A_{t^{\prime}} dt^{\prime}} = \sum_{k=0}^{\infty} \frac{1}{k!} \int\limits_{t_0}^{t} dt_{1} \int\limits_{t_0}^{t} dt_{2} \, ... \int\limits_{t_0}^{t} dt_{k} \, \mathcal{T} \left\{ A_{t_1} A_{t_2} \, ... \, A_{t_k} \right\} ,
\end{equation}
where $\mathcal{T}$ stands for ordering operator.

\subsection{ODE with periodic coefficient. Floquet theory}

In case of finite-dimensional $\bs$, in order to find a solution of ODE \eqref{eq:ODEstateTransitionOperator} one can apply the celebrated \emph{Floquet theorem}\cite{Floquet1883} (we formulate it in language of $U_{t,t_0}$):

\begin{theorem}[\textbf{Floquet}]

Let be given a mapping $\real \ni t \longrightarrow A_{t} \in \bx$, at least piecewise continuous and periodic with period $T$ such that $A_{t+T} = A_{t}$ for every $t\in\real$. Then, if $U_{t,t_0}$ is a \emph{fundamental solution} of ordinary differential equation $\frac{d}{dt}f(t) = A_{t} f(t)$, then $U^{\prime}_{t,t_0} = U_{t+T,t_0}$ is also a fundamental solution.
\end{theorem}

\begin{proof}
Set $\xi (t) = t+T$. Differentiating directly with respect to $t$, we obtain
\begin{equation}
	\frac{d}{dt}U_{\xi(t),t_0} = \frac{d\xi(t)}{dt} \frac{dU_{\xi(t),t_0}}{d\xi(t)} = A_{\xi(t)} U_{\xi(t),t_0} = A_{t+T} U_{t+T,t_0} = A_{t} U_{t+T,t_0}
\end{equation}
so $U_{t+T,t_0}$ is a fundamental solution.
\end{proof}

One can also see that if $A_{t}$ is periodic, then $U_{t,t_0}$ is invariant with respect to translation by $nT$, i.e. $U_{t,t_0} = U_{t+nT, t_0 + nT}$, $n\in\integer$. 
To prove this, write $U_{t+nT,t_0 +nT}$ as the ordered exponential \eqref{eq:OrderedExponential} and put $t_{k}^{\prime} = t_{k} - nT$ instead of $t_{k}$; then it easily follows from periodicity of $A_{t}$ that $\mathcal{T}\left\{ A_{t_{1}^{\prime}+nT} \, ... \, A_{t_{k}^{\prime}+nT} \right\} = \mathcal{T}\left\{ A_{t_{1}^{\prime}} \, ... \, A_{t_{k}^{\prime}} \right\}$ which yields $U_{t+nT,t_0 +nT} = U_{t,t_0}$. One also has $U_{t+T,t_0} = U_{t,t_0} U_{t_0 + T,t_0}$ where $U_{t_0 + T,t_0}$ is called \emph{a monodromy operator}. Again, proof here is simple; from Chapman -- Kolmogorov properties of $U_{t,t_0}$ as fundamental solution it easily follows that $U_{t+T,t_0} = U_{t,t_0 - T} = U_{t,t_0}U_{t_0,t_0-T} = U_{t,t_0} U_{t_0+T,t_0}$.

\begin{definition}[\textbf{Floquet representation}]
By \emph{Floquet representation of order $m$} we mean such a triple $(P,B,m) \in \bx \times \bx \times \nat$, that $P$ is periodic with period $mT$, $P_{t+mT,t_0} = P_{t,t_0}$ and $U_{t,t_0}$ may be expressed as $U_{t,t_0} = P_{t,t_0} e^{B(t-t_0)}$.
\end{definition}

\begin{remark}\label{remark:FloquetRepresentationOrder}
In traditional framework of finite-dimensional Banach space, one can prove that a sufficient condition for Floquet representation of order $m$ of $U_{t,t_0}$ to exist is the existence of logarithm of $m$-th power of monodromy operator\cite{MasseraSchaffer59,Schaffer1964}, i.e. if there exists a constant operator $B \in \bx$ such that $U_{t_0+T,t_0} = e^{mBT}$. Moreover, if $\bs$ is complex, there exists a Floquet representation of $U_{t,t_0}$ of order at most 1, and if $\bs$ is real, of order at most 2.
\end{remark}

\begin{corollary}
If $\bs$ is a finite-dimensional, complex Banach space and monodromy operator has a logarithm such that $U_{t_0+T,t_0} = e^{BT}$, there exists a Floquet representation (of order 1)  such that $U_{t,t_0} = P_{t,t_0} e^{B(t-t_0)}$ and $P$ is periodic with period $T$.
\end{corollary}

\begin{proof}
This claim is easy to prove by construction. Define $P_{t,t_0} = U_{t,t_0} e^{-B(t-t_0)}$. It follows directly from remark \ref{remark:FloquetRepresentationOrder}, that
\begin{equation}
	P_{t+T,t_0} = U_{t+T,t_0} e^{-B(t+T-t_0)} = U_{t,t_0} e^{BT} e^{-BT} e^{-B(t-t_0)} = U_{t,t_0} e^{-B(t-t_0)} = P_{t,t_0}.
\end{equation}
Therefore $P_{t,t_0}$ is periodic and $U_{t,t_0} = P_{t,t_0} e^{B(t-t_0)}$.
\end{proof}
\noindent Let us now assume that $\bs$ is a Hilbert space $\hilbert$. The crucial requirement that must be fulfilled in order to justify the Floquet approach in setting of general Hilbert space, is the existence of logarithm of monodromy operator\cite{Schaffer1964}. When $\hilbert$ is complex and finite-dimensional, the sufficient and necessary condition for existence of logarithm of operator $A$, represented by matrix from algebra $\mathcal{M}_{\dimension{\hilbert}}(\complex)$, is that $A^{-1}$ exists. In case of infinite-dimensional $\hilbert$, there may be no Floquet representation at all because logarithm of $A$ may not exist even if $A$ is invertible\cite{MasseraSchaffer59}. Therefore we will make few simplifying assumptions in order to proceed with derivation of appropriate dynamical semigroup at the very end. Namely, we assume, that:
\begin{enumerate}
\item logarithm of $U_{t_0+T,t_0}$ exists, i.e. there is a constant, bounded $B\in\bh$ such that $U_{t_0+T,t_0} = e^{BT}$, which implies, that $U_{t+T,t_0} = U_{t,t_0} e^{BT}$,
\item $B$ is a normal operator with spectrum $\sigma(B)$ being pure-point, i.e. it satisfies an eigenequation $B \phi_{k} = \mu_{k} \phi_{k}$ where $\{\phi_{k}\}$ is an orthonormal basis in $\hilbert$ and numbers $\mu_k \in \sigma(B)$ are called \emph{Floquet exponents} and, in principle, can be complex.  
\end{enumerate}
Obviously, monodromy operator $e^{BT}$ is diagonalized by the same set of vectors and $\sigma(e^{BT}) = \{e^{\mu_{k}T}, k \in \nat\}$.

\begin{proposition}\label{theorem:psiBasis}
Define new functions $\psi_{k}(t) = U_{t,t_0} \phi_k$ and $\phi_k (t) = e^{-\mu_k t}\psi_k (t)$. Then, set $\{\psi_{k}(t) \}$ is a basis in $\hilbert$ and functions $\phi_k (t)$ are periodic.
\end{proposition}

\begin{proof}

We assume that $\{\phi_k\}$ is a basis in $\hilbert$. Since $U_{t,t_0}^{-1}$ exists for all $t\in\mathbb{R}$, it is a bijection and set $\{\phi_k\}$ is isomorphic to $\{\psi_{k}(t)\}$. Therefore it is enough to prove the linear independence of $\{\psi_{k}(t)\}$. We obtain $0 = \sum_{k} \alpha_k \psi_{k}(t) = \sum_{k} \alpha_k U_{t,t_0} \phi_k = U_{t,t_0} \left( \sum_k \alpha_k \phi_k\right)$, which gives $\sum_k \alpha_k \phi_k = 0$, since it must hold for any $t$. But it immediately implies $\alpha_k = 0$ for all $k$ since $\{\phi_k\}$ is a basis and therefore linearly independent set. The conclusion is that $\{\psi_{k}(t)\}$ is also linearly independent.
\par
One can check that equality $\psi_k (t+T) = e^{\mu_k T} \psi_k (t)$ holds. Recall that $U_{t+T,t_0} = U_{t,t_0} e^{BT}$. For any $\psi_k (t) = U_{t,t_0} \phi_k$ we have
\begin{align}
	\psi_{k} (t+T) &= U_{t+T,t_0} \phi_k = U_{t,t_0} e^{BT} \phi_k = e^{\mu_k T} U_{t,t_0} \phi_k = e^{\mu_k T} \psi_k(t)
\end{align}
since $e^{BT} \phi_k = e^{\mu_k T} \phi_k$. Hence, calculating directly, we get
\begin{align}
	\phi_k (t+T) &= \psi_k(t+T) e^{-\mu_k (t+T)} = e^{\mu_k T} \psi_k (t) e^{-\mu_k (t+T)} = e^{-\mu_k t} \psi_k (t) = \phi_k (t)
\end{align}
which completes the proof.
\end{proof}

\section{Open Quantum Systems}
\label{sec:OQS}

\subsection{Preliminaries: periodic Hamiltonians and unitary propagator}
\label{subsec:QMperiodicH}

The starting point of our approach which will lead eventually to dynamical semigroups, is the well-known \schr equation, $\dot{\psi}(t) = -i H(t) \psi(t)$, where $\psi (t) \in \hilbert$, $\| \psi(t) \| = 1$ and $H(t) = H(t)^{\adj}$ is a self-adjoint, time-dependent Hamiltonian .
\par
We make a crucial assumption of periodicity of $H(t)$, namely $H(t+nT) = H(t)$ for $n \in \integer$, so the solution of \schr equation may be found by applying the Floquet theorem as outlined in previous sections, by putting $A_{t} = -i H(t)$ as stated in \eqref{ODEs}. State transition operator $U_{t,t_0}$ is now understood as a quantum-mechanical \emph{unitary propagator},
\begin{equation}\label{eq:UnitaryPropagatorDysonSeries}
	U_{t,t_0} = \mathcal{T} \exp{\left\{ -i \int\limits_{t_0}^{t} H(t^{\prime}) dt^{\prime} \right\}},
\end{equation}
which obviously satisfies all the properties of state transition operator, namely $U_{t,s}U_{s,t_0} = U_{t,t_0}$ for $t_0 \leqslant s \leqslant t$, $U_{t_0, t_0} = \id{\hilbert}$ and, naturally, $\frac{d}{dt}U_{t,t_0} = -i H(t) U_{t,t_0}$. Under the periodicity assumption of $H(t)$, we have a nice following proposition:

\begin{proposition}\label{prop:EvolutionFloquetTheorem}

There exist a constant, self-adjoint operator $\hav \in \bh$ and periodic $P_{t,t_0} \in \bh$ such that the evolution operator $U_{t,t_0}$ corresponding to $H(t)$ satisfies:
\begin{enumerate}
\item $U_{t+T,t_0} = U_{t,t_0}e^{-i\hav T}$,
\item $U_{t,t_0} = P_{t,t_0} e^{-i\hav (t-t_0)}$.
\end{enumerate}
\end{proposition}

\begin{proof}
\hfill
\vskip\baselineskip
\noindent\emph{Ad 1.} We assumed previously, that logarithm of monodromy operator exists, i.e. $U_{t_0+T,t_0} = e^{BT}$. Since $U_{t,t_0}$ is unitary for all $[t_0, t] \subset \real$, unitarity of $e^{BT}$ is assured; from this we conclude that $B$ is skew-adjoint, i.e. it is of a form $B = -i\hav$, where $\hav = \hav^{\adj}$.
\vskip\baselineskip
\noindent\emph{Ad 2.} This is actually a straightforward corollary coming from remark \ref{remark:FloquetRepresentationOrder}, which defines a Floquet representation of $U_{t,t_0}$ such that $P_{t,t_0}$ is periodic and $B = -i\hav$. We have $\|e^{-i\hav(t-t_0)}\| = 1$, and from unitarity of $U_{t,t_0}$ comes also the unitarity, and hence boundedness, of $P_{t,t_0}$.
\end{proof}

\noindent This newly introduced, self-adjoint operator $\hav$ will be named \emph{averaged Hamiltonian} and is of crucial importance for our approach. $\hav$ is diagonalized by vectors $\phi_{k}$ such that
\begin{equation}\label{eq:HbarSpectralDecomp}
	\hav \phi_{k} = \epsilon_{k} \phi_{k}
\end{equation}
where eigenvalues $\epsilon_{k} \in \real$ will be called \emph{Bohr-Floquet quasienergies}. Set of its eigenvectors provides a basis, called \emph{Floquet basis}. In fact, actual values of quasienergies (or, more precisely, differences between them) will give us the opportunity to construct an appropriate Markovian master equation, leading to derivation of quantum dynamical semigroup.

\subsection{General mathematical framework}

For the last few decades, a very rich and extensive theory has been under development to make possible the description of open systems, described in a language of \emph{density operator} $\dm{t} \in \traceclass \subset \kh$ where $\traceclass$ and $\kh$ respectively denote spaces of all trace-class and compact operators over $\hilbert$. It satisfies $\dm{t}^{\adj} = \dm{t}$ (self-adjointness), $\iprod{\psi}{\dm{t}\psi} > 0$ for all non-zero $\psi \in \hilbert$ (positivity) and $\| \dm{t} \|_{1} = \tr{\sqrt{\dm{t}^{\adj}\dm{t}}} = \tr{\dm{t}} = 1$ (trace norm one). General approach to this theory had been studied extensively since 1960's and today may be found in many sources \cite{AlickiLendi,BreuerPetruccione2002,Davies1976,AttalJoyePillet2006}. We focus on equation of motion for $\dm{t}$, expressed in general form as \cite{AlickiLendi,Alicki79}
\begin{equation}\label{eq:GeneralMasterEquation}
	\frac{d}{dt}\dm{t} = \lgen{t}{\dm{t}},
\end{equation}
where map $\lgens{t}$ on $\traceclass$ is called a \emph{Lindbladian} and is continuous in strong operator topology. We emphasize, that we only take linear $\lgens{t}$ into account. For each value of $t$, $\lgens{t}$ is, by Hille-Yosida theorem, considered as closed, infinitesimal generator of \czero $\Psi_{\tau}^{(t)} : \tau \geqslant 0$ of contracting maps on $\traceclass$ of a form $\Psi_{\tau}^{(t)} = e^{\tau\lgens{t}}$ (see e.g. ref. \onlinecite{AlickiLendi,Davies1976,GKS76} for further details). We have $\Psi_{0}^{(t)} = \id{\traceclass}$ and this semigroup is closed under binary operation $\circ$ understood as composition, fulfilling a general composition rule (semigroup property) $\Psi_{\tau_1}^{(t)} \circ \Psi_{\tau_2}^{(t)} = \Psi_{\tau_1 + \tau_2}^{(t)}$.
\par
A map $T$ is called $k$-positive if $T\otimes \id{\complex^{k}}$ acting on $\traceclass \otimes \mathscr{M}_{k}(\complex)$ is also positive for all $f \in \traceclass$ and $\varphi \in \mathscr{M}_{k}(\complex)$. If additionally $T$ is $k$-positive for any natural $k$, $T$ is called completely positive (CP). It is then required that any $\Psi_{\tau}^{(t)}$ is CP and moreover, that it preserves a trace norm, $\| \Psi_{\tau}^{(t)}(f)\|_{1} = \| f \|_{1}$ for any $f \in \traceclass$. Maps satisfying those two conditions are called completely positive and trace preserving (CPTP). The last property implies, that $\tr{\lgen{t}{\dm{t}}} = 0$.
\par
\czero of such strongly differentiable maps is called a \emph{quantum dynamical semigroup}, and its members $\Psi_{\tau}^{(t)}$ are commonly referred as \emph{quantum dynamical maps}. If $\lgens{t}$ is bounded, every $\Psi_{\tau}^{(t)}$ can be defined by a familiar power series expansion, converging in uniform operator topology since $\| e^{\tau\lgens{t}} \| \leqslant e^{\tau \| \lgens{t} \|} < \infty$. In this case $\Psi_{\tau}^{(t)}$ is even uniformly continuous.
\par
In general situation, when $\lgens{t}$ is not constant throughout a whole evolution, which is also of most interest for us, the mathematical structure standing behind reduced dynamics is much less regular than in time-independent case. However, there are still some interesting features, worth pointing out. Consider for simplicity, that $\| \lgens{t} \| < \infty$. Then, integrating \eqref{eq:GeneralMasterEquation} with initial condition $\dm{t_0}$ brings up a time-ordered exponential formula,
\begin{align}
	\dm{t} &= \left( \sum_{n=0}^{\infty} \frac{1}{n!} \int\limits_{t_0}^{t} dt_{1} ... \int\limits_{t_0}^{t} dt_{n} \, \mathcal{T} \left\{\lgens{t_1} \circ ... \circ \lgens{t_n}\right\} \right)(\dm{t_0}) \\
	&= \left(\mathcal{T} \exp{\int\limits_{t_0}^{t} \lgens{t^{\prime}} \, dt^{\prime}} \right) (\rho_{t_0}) = \qdm{t}{t_0}{\dm{t_0}}, \nonumber
\end{align}
where $\mathcal{T}$ denotes time ordering. Here we have introduced a quantum dynamical map $\qdms{t}{t_0}$,
\begin{equation}\label{eq:LambdaGeneral}
	\qdms{t}{t_0} = \mathcal{T} \exp{\int\limits_{t_0}^{t} \lgens{t^{\prime}} \, dt^{\prime}},
\end{equation}
which is required to be CPTP. On the other hand, consider again equation \eqref{eq:GeneralMasterEquation} but now with different initial condition $\dm{s}$ such that $\dm{s} = \qdm{s}{t_0}{\dm{t_0}}$ for some $t_0 \leqslant s$. Then we have $\dm{t} = \qdm{t}{s}{\dm{s}} = \qdm{t}{s}{\qdm{s}{t_0}{\dm{t_0}}}$, so in general, map $\qdms{t}{t_0}$ must fulfill a composition rule
\begin{equation}\label{eq:PropagatorCompRule}
	\qdms{t}{t_0} = \qdms{t}{\epsilon} \circ \qdms{\epsilon}{t_0}
\end{equation}
for any partition $t_0 \leqslant \epsilon \leqslant t$. A mapping $(t,s) \longmapsto \qdms{t}{s}$ is strongly continuous and differentiable with respect to $t$ and $s$ such that
\begin{subequations}
\begin{align}
	&\frac{d\qdms{t}{s}}{dt}  = \slim{h\searrow 0} \frac{\qdms{t+h}{t} - \id{\traceclass}}{h} \circ \qdms{t}{s}  = \lgens{t} \circ \qdms{t}{s}, \\
	&\frac{d \qdms{t}{s}}{ds} = \slim{h\searrow 0} \left(-\qdms{t}{s+h} \circ \frac{\qdms{s+h}{s} - \id{\traceclass}}{h}\right) = -\qdms{t}{s} \circ \lgens{s},
\end{align}
\end{subequations}
where $s-\mathrm{lim}$ denotes a limit in a sense of strong operator topology. By convention and for simplicity, one usually puts $t_0 = 0$ and provides a somehow simplified notation such that quantum dynamical map $\qdms{t}{0}$ gets replaced by $\Lambda_{t}$. In such a case, a map $\qdms{t}{s}$ for $s\neq 0$ is referred rather as \emph{propagator}, since it takes dynamical map at time $s$ to another dynamical map at time $t \geqslant s$, $\Lambda_{t} = \qdms{t}{s} \circ\Lambda_{s}$ which is a consequence of composition rule. We will occasionally use both names, propagator and dynamical map, interchangeably.
\vskip\baselineskip
For general, time-dependent Lindbladians, finding a corresponding propagator $\qdms{t}{t_0}$ may be a difficult task and \eqref{eq:LambdaGeneral} has purely formal meaning. Of course in case of constant $\mathcal{L}$, we have $\qdms{t}{t_0} = e^{(t-t_0)\mathcal{L}}$. A widely celebrated result by G. Lindblad, V. Gorini, A. Kossakowski and E. C. G. Sudarshan (see ref. \onlinecite{Lindblad76,GKS76}) presents a most general form of semigroup generator in case of constant, bounded $\mathcal{L}$ and separable $\hilbert$,
\begin{equation}\label{eq:MME}
	\mathcal{L}(\dm{t}) = -i\comm{\hams}{\dm{t}} + \sum_{j \in I} \left( V_{j} \dm{t} V_{j}^{\adj} - \frac{1}{2} \acomm{V_{j}^{\adj}V_{j}}{\dm{t}} \right)
\end{equation}
where $I$ is countable, $\hams = \hams^{\adj} \in \bh$ and $V_{j}, \sum_{j} V_{j}^{\adj}V_{j} \in \bh$.

\subsection{Time-dependent Lindbladians. Periodic Hamiltonians}
\label{subsec:MainResult}

\subsubsection{General framework}\label{subsect:GeneralFramework}

A mathematically rigorous approach by Davies \cite{Davies1976,Davies74} allows to find the exact form of $\lgen{}{\dm{t}}$ in a regime of \emph{weak coupling limit}, in terms of what is now known as Markovian master equation. Usually, one starts with some isolated, composite system $\Sigma$, which may be divided into two subsystems, $\mathrm{S}$ and $\mathrm{R}$, denoting system of interest and surrounding environment. They are described by Hilbert spaces $\hilberts$ and $\hilbertr$, respectively, such that $\mathscr{H}_{\Sigma} = \hilbertsr$.
It is often assumed that $\mathrm{S}$ is ``small'', i.e. either $\dim{\hilberts} < \infty$ or spectrum of its Hamiltonian $\sigma(H_{\mathrm{S}})$ is pure point, and that R is ``large'', i.e. $\dim{\hilbertr}$ is countably infinite. One considers a joint density operator $\sigma_{t} \in \mathscr{B}_{1}(\hilbertsr)$ which undergoes unitary evolution with respect to composite Hamiltonian $H_{\Sigma} = \hams\otimes\id{\hilbertr} + \id{\hilberts}\otimes\hamr + \hami$, where $\hams,\hamr,\hami \in \bh$ are self-adjoint Hamiltonians of system of interest, environment and interaction between S and R, respectively. Furthermore, it is assumed that $\hami$ takes a general form
\begin{equation}\label{eq:HIntGeneralForm}
	\hami = \lambda \sum_{\alpha \in I} S_{\alpha} \otimes R_{\alpha}
\end{equation}
with $I$ being countable, $S_{\alpha} : \hilberts \longrightarrow \hilberts$, $R_{\alpha} : \hilbertr \longrightarrow \hilbertr$, $S_{\alpha} = S_{\alpha}^{\adj}$, $R_{\alpha} = R_{\alpha}^{\adj}$ and $\lambda \in \real$ considered as ``small'' (up to energy scale) coupling constant. A usual derivation of time-independent semigroup generator makes use of interaction picture dynamics with respect to $\hams\otimes\id{\hilbertr} + \id{\hilberts}\otimes\hamr$. A detailed derivation for constant $\hams$, which is a core aspect of Davies' formulation, may be found in literature (see a detailed description in ref. \onlinecite{AlickiLendi,BreuerPetruccione2002,Davies74,Davies1976,GKS76,Lindblad76}). Here we are extending this formalism to the case of periodic $\hams(t)$.

\subsubsection{Time-dependent Lindbladian. Sketch of derivation}

\noindent In the following, we denote by $\dm{t} \in \tcsystem$ the reduced density operator of S, given by partial trace over $\hilbertr$, $\dm{t} = \mathrm{tr}_{\hilbertr}\{\sigma_{t}\}$. Let a mapping $t \longmapsto \dm{t}$ to be differentiable for all $t$ such that $\dm{t}$ satisfies a generalized, time-dependent Markovian master equation (with initial condition $\dm{t_0}$)
\begin{equation}\label{eq:TheoremMME}
	\frac{d}{dt}{\rho}_{t} = \lgen{t}{\dm{t}}.
\end{equation}
A specific form of $\lgens{t}$ was proposed by Alicki, Lidar and Zanardi in ref. \onlinecite{AlickiLidarZanardi2006} and mathematical validity of that approach will be proven. For consistency, we will accept some naming conventions used in the aforementioned paper. We formulate a main result as a following theorem:

\begin{theorem}
Let $\mathrm{S}$ be the open quantum system, weakly coupled to external reservoir $\mathrm{R}$ as elaborated in subsection \ref{subsect:GeneralFramework}. Let $\hams(t)$ denote a time-dependent, self-adjoint Hamiltonian of system of interest, which is periodic with period $T$, i.e. $\hams(t+nT) = \hams(t)$ for all $t\in\real$ and $n\in\integer$. Then, there exists a time-independent \czero generator $\iplgens$ and unitary map $\mathcal{U}_{t,t_0}$ such that the Markovian master equation \eqref{eq:TheoremMME} can be, under weak coupling limit, expressed by
	\begin{subequations}
	\begin{align}
		&\lgens{t} = -i \comm{\hams(t)}{\,\cdot\,} + \disss{t}, \label{eq:TheoremClaim1}\\
		&\disss{t} = \mathcal{U}_{t,t_0} \circ \iplgens \circ \mathcal{U}_{t,t_0}^{-1}, \label{eq:TheoremClaim2}
	\end{align}
	\end{subequations}
and a quantum dynamical map $\qdms{t}{t_0}$ generated by $\lgens{t}$ has a form
\begin{equation}\label{eq:QDMunderFloquetTheorem}
	\qdms{t}{t_0} = \mathcal{U}_{t,t_0}\circ e^{(t-t_0) \iplgens}.
\end{equation}
\end{theorem}
\noindent In the proof, we will present a derivation of appropriate Lindbladian $\lgens{t}$. We will show, by construction, that claimed maps $\iplgens$ and $\mathcal{U}_{t,t_0}$ exist and both the time-dependent Lindbladian $\lgens{t}$ and dynamical map $\qdms{t}{t_0}$ in the \schr picture are of proposed forms. It will be shown that $\lgens{t}$ has a similar, Lindblad-like structure, hence produces CP and trace preserving dynamical maps. Some various properties of $\lgens{t}$ and $\qdms{t}{t_0}$ will be shown.

\begin{proof}

We consider a composite dynamics of $\sigma_{t} \in \tcjoint$. It is convenient to switch to the interaction picture with respect to $H_{\Sigma}(t)= \hams(t) \otimes \id{\hilbertr} + \id{\hilberts}\otimes\hamr$. To achieve this, let us define a unitary map $\usupsgen{\Sigma}{t}{t_0}$ on $\mathscr{B}(\hilbertsr)$ such that for any $A$ it is given by equalities
\begin{align}
	\usupgen{\Sigma}{t}{t_0}{A} = U^{\Sigma}_{t,t_0} A (U^{\Sigma}_{t,t_0})^{-1},\qquad \ausupgen{\Sigma}{t}{t_0}{A} = (U^{\Sigma}_{t,t_0})^{-1} A U^{\Sigma}_{t,t_0}
\end{align}
where we define $U^{\Sigma}_{t,t_0}$ as
\begin{align}
	U^{\Sigma}_{t,t_0} = U_{t,t_0} \otimes U^{R}_{t,t_0}, \qquad (U^{\Sigma}_{t,t_0})^{-1} = U_{t,t_0}^{-1} \otimes (U^{R}_{t,t_0})^{-1},
\end{align}
$U^{R}_{t,t_0} = e^{-i\hamr (t-t_0)}$ is a free evolution of R and $U_{t,t_0}$ stands for a time-ordered analogue of unitary evolution operator of S from conventional approach,
\begin{equation}\label{eq:UnitaryPropagatorTimeOrdered}
	U_{t,t_0} = \mathcal{T} \exp{\left\{-i\int\limits_{t_0}^{t} H_{S}(t^{\prime}) \, dt^{\prime}\right\}}.
\end{equation}
For general Banach spaces $X$ and $Y$, we can define a tensor product of two maps $\mathcal{A} \in \mathscr{B}(X)$, $\mathcal{B} \in \mathscr{B}(Y)$ to be such a map $\mathcal{A} \otimes \mathcal{B}$, that $(\mathcal{A} \otimes \mathcal{B})(x \otimes y) = \mathcal{A}(x) \otimes \mathcal{B}(y)$ for $x \otimes y \in X \otimes Y$. Then it is evident, that for $x\otimes y \in \hilbertsr$ one can write
\begin{align}
	\usupgen{\Sigma}{t}{t_0}{x\otimes y} = \usupgen{\mathrm{S}}{t}{t_0}{x} \otimes \usupgen{\mathrm{R}}{t}{t_0}{y},
\end{align}
where two new unitary maps $\usupsgen{\mathrm{S}}{t}{t_0}$ and $\usupsgen{\mathrm{R}}{t}{t_0}$ are defined as follows,
\begin{equation}
	\usupgen{\mathrm{S}}{t}{t_0}{A} = U_{t,t_0} A U_{t,t_0}^{-1}, \label{eq:USsuperoperator}, \qquad \usupgen{\mathrm{R}}{t}{t_0}{A} = U_{t,t_0}^{\mathrm{R}} A (U_{t,t_0}^{\mathrm{R}})^{-1}.
\end{equation}
Corresponding inverse maps are obtained simply by replacing $U_{t,t_0}$ and $U^{\mathrm{R}}_{t,t_0}$ by $U_{t,t_0}^{-1}$ and $(U_{t,t_0}^{\mathrm{R}})^{-1}$. Transformations between \schr and interaction pictures can now be defined elegantly by applying those maps to operators such that one can replace $\sigma_{t}$ and $\hami$ by their interaction picture counterparts $\tilde{\sigma}_{t}$ and $\iphami{t}$,
\begin{align}
	&\tilde{\sigma}_{t} = \ausupgen{\Sigma}{t}{t_0}{\sigma_{t}},\\
	&\iphami{t} = \ausupgen{\Sigma}{t}{t_0}{\hami}.
\end{align}
We will now give a sketch of usual \emph{microscopic} derivation of Markovian master equation, comparable to the one which is present in most modern textbooks. Since it is assumed that $\sigma_{t}$ evolves unitarily in $\mathscr{B}_{1}(\hilbertsr)$, we obtain the interaction picture version of von Neumann equation
\begin{equation}\label{eq:vonNeumannIP}
	\frac{d}{dt} \tilde{\sigma}_{t} = -i \comm{\iphami{t}}{\tilde{\sigma}_{t}},
\end{equation}
which we can integrate to obtain a formal solution
\begin{equation}\label{eq:FormalSolution}
	\tilde{\sigma}_{t} = \tilde{\sigma}_{t_0} - i \int\limits_{t_0}^{t} \comm{\iphami{t^{\prime}}}{\tilde{\sigma}_{t^{\prime}}} \, dt^{\prime}.
\end{equation}
Substituting \eqref{eq:FormalSolution} back into \eqref{eq:vonNeumannIP} and performing a partial trace with respect to $\hilbertr$, one obtains
\begin{equation}
	\frac{d}{dt}\ipdm{t} = \frac{d}{dt}\ptr{\hilbertr}{\tilde{\sigma}_{t}} = - \int\limits_{t_0}^{t} \ptr{\hilbertr}{\comm{\iphami{t}}{\comm{\iphami{t^{\prime}}}{\sigma(t^{\prime})}}} \, dt^{\prime}
\end{equation}
valid, if $\ptr{\hilbertr}{\comm{\iphami{t}}{\sigma(t_0)}} = 0$. In order to proceed with derivation, one applies a series of approximations\cite{AlickiLendi,Davies1976,BreuerPetruccione2002} and, after some algebra, arrives at coarse-grained, Markovian master equation
\begin{equation}\label{eq:PrototypeMME}
	\frac{d}{dt} \ipdm{t} = -\int\limits_{t_0}^{\infty} \ptr{\hilbertr}{\comm{\iphami{t}}{\comm{\iphami{t-t^{\prime}}}{\ipdm{t}\otimes\omega}}} \, dt^{\prime}
\end{equation}
where $\omega = \ptr{\hilberts}{\tilde{\sigma}_{t}} \in \tcenv$ is considered as a constant density operator of the environment, $\comm{\omega}{\hamr}=0$. The obtained formula does not yet generate a CPTP map (as noted by Davies\cite{Davies74}, D\"{u}mke and Spohn \cite{DumkeSpohn1979}), and to achieve that one has to engage some variation of rotating wave approximation, effectively averaging over rapidly oscillating terms in \eqref{eq:PrototypeMME}. Therefore one wants to rewrite $\iphami{t}$ in such a way that it is possible to subtract from it some time-dependent, oscillating terms.
\vskip\baselineskip
In the original approach, system's Hamiltonian was considered constant and one could obtain, at the very end, proper semigroup structure of reduced dynamics. Here we present a different approach; we will rewrite $\usupsgen{\mathrm{S}}{t}{t_0}$, applying some of the results based on Floquet theory from previous section, to cover periodicity of $\hams(t)$. Namely, we will use the Floquet representation of unitary propagator $U_{t,t_0}$ to obtain plausible expression for semigroup generator $\iplgens$ in interaction picture and show that it still has a familiar Lindblad structure.
\par
First, it is easy to notice and was mentioned before, that $U_{t,t_0}$ defining $\usupsgen{\mathrm{S}}{t}{t_0}$ in \eqref{eq:USsuperoperator} satisfies all Chapman-Kolmogorov properties and obviously, constitutes an example of (principal) fundamental solution for ODE given in form of \schr equation. By proposition \ref{prop:EvolutionFloquetTheorem} there exists a self-adjoint, \emph{averaged Hamiltonian} $\hav$ and a periodic, unitary operator $P_{t,t_0}$ such that $U_{t,t_0} = P_{t,t_0} e^{-i \hav (t-t_0)}$. This lets us to conclude, that for $A \in \mathscr{B}(\hilberts)$
\begin{equation}
	\usupgen{\mathrm{S}}{t}{t_0}{A} = (\mathcal{P}_{t,t_0} \circ \bar{\mathcal{U}}_{t,t_0})(A)
\end{equation}
with appropriate maps $\mathcal{P}_{t,t_0}$ and $\bar{\mathcal{U}}_{t,t_0}$ defined by relations
\begin{subequations}
\begin{align}
	&\mathcal{P}_{t,t_0}(A) = P_{t,t_0} A P_{t,t_0}^{-1}, \qquad \mathcal{P}_{t,t_0}^{-1}(A) = P_{t,t_0}^{-1} A P_{t,t_0}, \\
	&\mathcal{\bar{U}}_{t,t_0}(A) = e^{-i \hav (t-t_0)} A e^{i \hav (t-t_0)}, \qquad \mathcal{\bar{U}}_{t,t_0}^{-1}(A) = e^{i \hav (t-t_0)} A e^{-i \hav (t-t_0)}.
\end{align}
\end{subequations}
Interaction picture form of \eqref{eq:HIntGeneralForm} is now defined via $\ausupsgen{\Sigma}{t}{t_0}$, so
\begin{align}\label{eq:IPHIexpansion}
	\iphami{t} &= \ausupgen{\Sigma}{t}{t_0}{\hami} = \lambda \sum_{\alpha} \left(\ausupsgen{\mathrm{S}}{t}{t_0} \otimes \ausupsgen{\mathrm{R}}{t}{t_0}\right)(S_{\alpha} \otimes R_{\alpha}) \\
	&= \lambda \sum_{\alpha} \ausupgen{\mathrm{S}}{t}{t_0}{S_{\alpha}} \otimes \ausupgen{\mathrm{R}}{t}{t_0}{R_{\alpha}} \nonumber \\
	&= \lambda \sum_{\alpha} \ips{\alpha}{t} \otimes \ipr{\alpha}{t} \nonumber
\end{align}
with explicit forms of time-dependent operators
\begin{subequations}
\begin{align}
	\ips{\alpha}{t} &= (\bar{\mathcal{U}}_{t,t_0}^{-1} \circ \mathcal{P}_{t,t_0}^{-1})(S_{\alpha}) = e^{i\hav(t-t_0)} ( P_{t,t_0}^{-1} S_{\alpha} P_{t,t_0} ) e^{-i\hav(t-t_0)}, \\
	\ipr{\alpha}{t} &= e^{i\hamr (t-t_0)} R_{\alpha} e^{-i\hamr(t-t_0)}.\label{eq:IPRdefinition}
\end{align}
\end{subequations}
$P_{t,t_0}$ is periodic with period $T$; therefore $P_{t,t_0}^{-1} S_{\alpha} P_{t,t_0}$ is periodic as well and we can decompose it into a Fourier series, i.e. there exists a net of operators $\{S_{\alpha}(q) : q \in \integer\}$ such that
\begin{equation}
	P_{t,t_0}^{-1} S_{\alpha} P_{t,t_0} = \sum_{q\in\integer} S_{\alpha}(q) e^{iq\Omega (t-t_0)},
\end{equation}
converging in strong operator topology, where $\Omega = 2\pi / T$. Recall a notion of a Floquet basis $\{\phi_{k}\}$, which diagonalized $\hav$ as in \eqref{eq:HbarSpectralDecomp}, i.e. $\hav\phi_{k} = \epsilon_{k}\phi_{k}$. Let $P(\phi_{k})$ denote an orthogonal, self-adjoint projection. Define
\begin{equation}
	S_{\alpha}(\omega,q) = \sum_{\{\epsilon_{k}-\epsilon_{k^{\prime}} = \omega\}} P(\phi_{k}) S_{\alpha}(q) P(\phi_{k^{\prime}}),
\end{equation}
where the sum is performed over these indices $k$ and $k^{\prime}$, such that $\epsilon_{k}-\epsilon_{k^{\prime}} = \omega$. Here we introduce a new set $\{\omega = \epsilon_{k} - \epsilon_{l}\}$ of differences between eigenvalues of averaged Hamiltonian which we call the \emph{Bohr-Floquet quasifrequencies}. Summing $S_{\alpha}(\omega,q)$ over all $\omega$ is equivalent to taking a sum over all possible pairs $(\phi_{k},\phi_{k^{\prime}})$, since $\omega$ denotes an equivalence class of eigenvectors; therefore we have $\sum_{\omega} S_{\alpha}(\omega,q) = \sum_{k,k^{\prime}} P(\phi_{k}) S_{\alpha}(q) P(\phi_{k^{\prime}}) = S_{\alpha}(q)$ because of completeness. Moreover, by direct computation it is also easy to check that
\begin{equation}\label{eq:CommutSalpha}
	\comm{\hav}{S_{\alpha}(\omega,q)} = \omega S_{\alpha}(\omega,q), \qquad \comm{\hav}{S_{\alpha}(\omega,q)^{\adj}} = -\omega S_{\alpha}(\omega,q)^{\adj}
\end{equation}
which leads to expression
\begin{align}
\bar{\mathcal{U}}^{-1}_{t,t_0}(S_{\alpha}(q)) &= e^{i\hav(t-t_0)} S_{\alpha}(q) e^{-i\hav(t-t_0)} \\
&= \sum_{\{\omega\}} S_{\alpha}(\omega,q) e^{i\omega (t-t_0)} = \sum_{\{\omega\}} S_{\alpha}(\omega,q)^{\adj} e^{-i\omega (t-t_0)}. \nonumber
\end{align}
Putting everything together, we have a following decomposition of $S_{\alpha}$ in the interaction picture
\begin{align}\label{eq:IPSexpansion}
	&\ips{\alpha}{t} = \sum_{q\in\integer}\sum_{\{\omega\}} S_{\alpha}(\omega,q) e^{i(\omega+q\Omega)(t-t_0)}, \\
	&S_{\alpha}(\omega,q) = \frac{1}{T} \int\limits_{-\frac{T}{2}+t_0}^{\frac{T}{2}+t_0} \ips{\alpha}{t} e^{-i(\omega+q\Omega)(t-t_0)} \, dt 
\end{align}
with additional property $S_{\alpha}(\omega, q)^{\adj} = S_{\alpha}(-\omega,-q)$. The expression $\omega + q\Omega$ may be understood as \emph{shifted quasifrequency} of $q$-th higher mode and is simply obtained by extending a set of Bohr-Floquet quasifrequencies. The hallmark of Markovian approximation is explicitly present here in the assumption, that the typical time scale $\tau_{\mathrm{S}}$ of intrinsic evolution of system of interest, given now by estimate $\tau_{\mathrm{S}} \sim \max_{\omega\neq\omega^{\prime},m\in\integer}\{|\omega-\omega^{\prime}+m\Omega|^{-1}\}$, is large, compared to typical relaxation time $\tau_{\mathrm{R}}$ and characteristic time $\tau_{\mathrm{B}}$, during which reservoir correlation functions decay. Once we have a proper decomposition of $\ips{\alpha}{t}$, the remaining derivation then follows the same path as in original Davies approach; by integrating over sufficiently large time, $t - t_0 \gg \max_{\omega\neq\omega^{\prime},m\in\integer}\{|\omega-\omega^{\prime}+m\Omega|^{-1}\}$, markovianity can be justified (further details may be found in appendix \ref{app:MME}). In the end, we arrive at the following expression for semigroup generator $\iplgens$ in the interaction picture,
\begin{subequations}
	\begin{align}
		&\iplgen{\ipdm{t}} = -i\comm{\delta H}{\ipdm{t}} + \ipdissprime{\ipdm{t}},\\
		&\ipdissprime{\ipdm{t}} = \sum_{\alpha\beta}\sum_{\{\omega\}}\sum_{q\in\integer} \gamma_{\alpha\beta}(\omega + q\Omega)\left( S_{\beta}(\omega,q) \, \ipdm{t} \, S_{\alpha}(\omega,q)^{\adj} - \frac{1}{2} \acomm{S_{\alpha}(\omega,q)^{\adj}S_{\beta}(\omega,q)}{\ipdm{t}} \right), \label{eq:LinteractionPicture}
	\end{align}
\end{subequations}
which is time-independent and resembles exactly the one presented by Alicki \emph{et al}. Prime in $\tilde{\mathcal{D}}^{\prime}$ indicates, that all terms expressing a Lamb-shift corrections due to influence of R on S were subtracted from dissipator and included into $\delta H$, commonly called a Lamb-shift Hamiltonian, given by formula
\begin{equation}
	\delta H = \lambda^{2}\sum_{\alpha\beta}\sum_{\{\omega\}}\sum_{q\in\integer}  \sigma_{\alpha\beta} (\omega+q\Omega) S_{\alpha}(\omega,q)^{\adj}S_{\beta}(\omega,q)
\end{equation}
with $[\sigma_{\alpha\beta}(\omega+q\Omega)]_{\alpha\beta}$ being hermitian matrix. Real-valued functions $\gamma_{\alpha\beta}$ are known as \emph{reservoir spectral density functions}. $\iplgens$ generates a one-parameter \czero $\ipqdms{t}{t_0} : t \geqslant t_0$ of CPTP maps given by
\begin{equation}
	\ipqdm{t}{t_0}{\ipdm{t_0}} = e^{(t-t_0)\iplgens}(\ipdm{t_0}),
\end{equation}
associative semigroup operation $\circ$ is understood as a composition $\ipqdms{t}{t_0} \circ \ipqdms{t^{\prime}}{t_{0}^{\prime}} = \ipqdms{t+t^{\prime}}{t_0 + t_{0}^{\prime}}$ and neutral element is simply $\ipqdms{t}{t} = \id{\traceclass}$.
\vskip\baselineskip
\par
If it is desired that $t_0 = 0$, $\gamma_{\alpha\beta}$ is given as Fourier transform of reservoir autocorrelation function,
\begin{equation}\label{eq:SpectralDensityFunctions_t0}
	\gamma_{\alpha\beta}(x) = \int\limits_{-\infty}^{\infty} e^{-ixt^{\prime}} \tr{\dmr \ipr{\alpha}{t}\ipr{\beta}{t-t^{\prime}}} \, dt^{\prime} = \int\limits_{-\infty}^{\infty} e^{-ixt^{\prime}} \tr{\dmr \ipr{\alpha}{t^{\prime}}\ipr{\beta}{0}} \, dt^{\prime}.
\end{equation}
Refer to appendix \ref{app:MME} for further details. Having a generator in interaction picture, we may return to original, \schr picture by applying $\usupsgen{\mathrm{S}}{t}{t_0}$ such that
\begin{equation}
	\dm{t} = \usupgen{\mathrm{S}}{t}{t_0}{\ipdm{t}}
\end{equation}
and $\ipdm{t} = \ipqdm{t}{t_0}{\ipdm{t_0}}$. Since we have $\dm{t_0} = \ipdm{t_0}$, we have constructed a proper propagator of a form $\qdms{t}{t_0} = \usupsgen{\mathrm{S}}{t}{t_0} \circ e^{(t-t_0)\iplgens}$ such that
\begin{equation}
	\dm{t} = \qdm{t}{t_0}{\dm{t_0}} = \usupgen{\mathrm{S}}{t}{t_0}{e^{(t-t_0)\iplgens}(\dm{t_0})}, \qquad \dm{t_0} = \ipdm{t_0}
\end{equation}
which proves claim \eqref{eq:QDMunderFloquetTheorem}. Of course we put explicitly $\usupsgen{\mathrm{S}}{t}{t_0} = \usups{t}{t_0}$. Differentiating $\dm{t}$ we get
\begin{align}\label{eq:derivative}
	\frac{d}{dt} \dm{t} &= \slim{h\searrow 0} \frac{\usup{t+h}{t_0}{\ipdm{t+h}} - \usup{t}{t_0}{\ipdm{t}}}{h} = \slim{h\searrow 0} \frac{\usup{t+h}{t_0}{\ipdm{t+h}-\ipdm{t} + \ipdm{t}} - \usup{t}{t_0}{\ipdm{t}}}{h} \\
	&= \slim{h\searrow 0} \frac{\usup{t+h}{t_0}{\ipdm{t}} - \usup{t}{t_0}{\ipdm{t}}}{h} \, + \, \slim{h\searrow 0} \, \usups{t+h}{t_0}\left( \frac{\ipdm{t+h}-\ipdm{t}}{h} \right) \nonumber \\
	&= \frac{\partial \usups{t}{t_0}}{\partial t}(\ipdm{t}) + \usups{t}{t_0}\left( \frac{d\ipdm{t}}{dt}\right), \nonumber
\end{align}
which easily comes from theorem concerning limits of composite continuous functions from analysis. The partial derivative of $\usups{t}{t_0}$ is a map defined by
\begin{equation}
	\frac{\partial \usups{t}{t_0}}{\partial t}(A) = \slim{h\searrow 0} \frac{\usup{t+h}{t_0}{A} - \usup{t}{t_0}{A}}{h} = \slim{h\searrow 0} \frac{U_{t+h,t_0}AU_{t+h,t_0}^{-1} - U_{t,t_0}A U_{t,t_0}^{-1}}{h}.
\end{equation}
Since $\usup{t+h}{t_0}{A}$ is smooth with respect to $h$ for $A \in \bh$, we can expand it into Maclaurin series near $h = 0$,
\begin{equation}
	\usup{t+h}{t_0}{A} = U_{t,t_0} A U_{t,t_0}^{-1} + h\frac{\partial}{\partial t} (U_{t,t_0} A U_{t,t_0}^{-1}) + \mathcal{O}(h^{2}),
\end{equation}
converging uniformly, and this yields
\begin{align}
	\frac{\partial \usups{t}{t_0}}{\partial t}(A) &= \frac{\partial}{\partial t} (U_{t,t_0} A U_{t,t_0}^{-1}) + \slim{h\searrow 0} \frac{\mathcal{O}(h^2)}{h} = \frac{\partial}{\partial t} (U_{t,t_0} A U_{t,t_0}^{-1})
\end{align}
as $\mathcal{O}(h^{2})/h \to 0 $ as $h \to 0$. Performing the differentiation, this naturally yields, as $\hams(t)$ commutes with $U_{t,t_0}$,
\begin{equation}
	\frac{\partial \usups{t}{t_0}}{\partial t}(A) = -i \hams(t) U_{t,t_0} A U_{t,t_0}^{-1} + i U_{t,t_0} A U_{t,t_0}^{-1} \hams(t) = -i \comm{\hams(t)}{\usup{t}{t_0}{A}}.
\end{equation}
The last term in \eqref{eq:derivative} is
\begin{equation}
	\usups{t}{t_0}\left( \frac{d\ipdm{t}}{dt}\right) = \usup{t}{t_0}{\iplgen{\ausup{t}{t_0}{\dm{t}}}}
\end{equation}
which, after putting $A = \ausup{t}{t_0}{\dm{t}}$, gives the proposed equation and one has
\begin{equation}\label{eq:RhoTderivative}
	\frac{d}{dt} \dm{t} = \lgen{t}{\dm{t}} = -i \comm{\hams(t)}{\dm{t}} + (\usups{t}{t_0}\circ\iplgens\circ\ausups{t}{t_0})(\dm{t})
\end{equation}
which proves remaining claims \eqref{eq:TheoremClaim1} and \eqref{eq:TheoremClaim2} and completes the construction.
\end{proof}

\noindent Matrix $[\gamma_{\alpha\beta}(x)]_{\alpha\beta}$ may be diagonalized by suitable unitary transformation and by doing so, one can obtain a diagonal form of dissipator,
\begin{equation}
	\ipdissprime{\ipdm{t}} = \sum_{\alpha}\sum_{\{\omega\}}\sum_{q\in\integer} \gamma_{\alpha}(\omega + q\Omega)\left( S_{\alpha}(\omega,q) \, \ipdm{t} \, S_{\alpha}(\omega,q)^{\adj} - \frac{1}{2} \acomm{S_{\alpha}(\omega,q)^{\adj}S_{\alpha}(\omega,q)}{\ipdm{t}} \right)
\end{equation}
with $\gamma_{\alpha}(x)$ being an abbreviation of $\gamma_{\alpha\alpha}(x)$. Moreover, one can put the obtained equation into a form
\begin{align}
	\lgen{t}{\dm{t}} &= -i \comm{H_{\mathrm{phys.}}(t)}{\dm{t}} + (\usups{t}{t_0}\circ\ \tilde{\mathcal{D}}^{\prime} \circ\ausups{t}{t_0})(\dm{t}) = \\
	&= -i \comm{H_{\mathrm{phys.}}(t)}{\dm{t}} + \mathcal{D}^{\prime}_{t}(\dm{t}), \nonumber
\end{align}
where $\mathcal{D}^{\prime}_{t} = \usups{t}{t_0}\circ\ \tilde{\mathcal{D}}^{\prime} \circ\ausups{t}{t_0}$ and $H_{\mathrm{phys.}}(t) = \hams(t)+\delta H$ may be understood as ``physical'' Hamiltonian, containing all Lamb-like corrections due to environmental influence. By suitable renormalization procedure sometimes one replaces $\hams(t)$ with $H_{\mathrm{phys.}}(t)$ at the very beginning of calculations such that $\delta H$ can be removed and $\mathcal{D}^{\prime}_{t}$ is identified with $\disss{t}$.
\vskip\baselineskip
\noindent For practical reasons, let us define another map $\mathcal{F}$ on $\bh$ by
\begin{subequations}
\begin{align}
	&\mathcal{F} (A) = U_{t_0+T,t_0} \, A \, U_{t_0+T,t_0}^{-1} = e^{-i\hav T} \, A \, e^{i \hav T}, \\
	&\mathcal{F}^{-1} (A) = U_{t_0+T,t_0}^{-1} \, A \, U_{t_0+T,t_0} = e^{i\hav T} \, A \, e^{-i \hav T}.
\end{align}
\end{subequations}

\begin{proposition}[\textbf{covariance property}]\label{prop:CovarianceProperty}
It holds, that $\comm{\mathcal{F}}{\iplgens} = \comm{\mathcal{F}^{-1}}{\iplgens} = 0$, where Lie bracket of maps $\mathcal{A}, \mathcal{B}$ on $\bh$ is defined naturally as a commutator, i.e. \\ $\comm{\mathcal{A}}{\mathcal{B}} = \mathcal{A}\circ\mathcal{B} - \mathcal{B}\circ\mathcal{A}$.
\end{proposition}

\begin{proof}
One needs to show that $\mathcal{F}(\iplgen{A}) = \iplgen{\mathcal{F}(A)}$ for $A \in \bh$. Since we have
\begin{equation}
	\mathcal{F}(S_{\alpha}(\omega,q)) = e^{-i\omega T} S_{\alpha}(\omega,q),\qquad \mathcal{F}(S_{\alpha}(\omega,q)^{\adj}) = e^{i\omega T} S_{\alpha}(\omega,q)^{\adj},
\end{equation}
which comes from \eqref{eq:CommutSalpha}, and obviously
\begin{equation}
	e^{-i\hav T} \, AB \, e^{i\hav T} = e^{-i\hav T} \, A \, e^{i\hav T} \, e^{-i\hav T} \, B \, e^{i\hav T},
\end{equation}
one obtains
\begin{equation}
	\mathcal{F}_{t_0}(S_{\alpha}(\omega,q)^{\adj}S_{\beta}(\omega,q)) = S_{\alpha}(\omega,q)^{\adj}S_{\beta}(\omega,q).
\end{equation}
Moreover, it easily follows that $\mathcal{F}(\comm{A}{B}) = \comm{\mathcal{F}(A)}{\mathcal{F}(B)}$ and $\mathcal{F}(\acomm{A}{B}) = \acomm{\mathcal{F}(A)}{\mathcal{F}(B)}$, which yields
\begin{subequations}
	\begin{align}
		&\mathcal{F}(\comm{S_{\alpha}(\omega,q)^{\adj}S_{\beta}(\omega,q)}{A}) = \comm{S_{\alpha}(\omega,q)^{\adj}S_{\beta}(\omega,q)}{\mathcal{F}(A)}, \\
		&\mathcal{F}(\acomm{S_{\alpha}(\omega,q)^{\adj}S_{\beta}(\omega,q)}{A}) = \acomm{S_{\alpha}(\omega,q)^{\adj}S_{\beta}(\omega,q)}{\mathcal{F}(A)}, \\
		&\mathcal{F}(S_{\alpha}(\omega,q) A S_{\beta}(\omega,q)^{\adj}) = S_{\alpha}(\omega,q) \mathcal{F}(A) S_{\beta}(\omega,q)^{\adj} .
	\end{align}
\end{subequations}
Schematically, $\iplgens$ takes a form $\iplgen{\ipdm{t}} = \sum S_a \ipdm{t} S_{b}^{\adj} - \frac{1}{2} \acomm{S_{b}^{\adj}S_{a}}{\ipdm{t}}$. Applying $\mathcal{F}$ to $\iplgen{\ipdm{t}}$ is, by above identities, equivalent to $\sum S_{a} \mathcal{F}(\ipdm{t}) S_{b}^{\adj} - \frac{1}{2} \acomm{S_{b}^{\adj}S_{a}}{\mathcal{F}(\ipdm{t})} = \iplgen{\mathcal{F}(\ipdm{t})}$ as can be easily checked by direct computation; therefore $\mathcal{F}\circ\iplgens - \iplgens\circ\mathcal{F} = 0$. Applying $\mathcal{F}^{-1}$ on both sides of this equality, from left and from right, also yields that $\mathcal{F}^{-1} \circ \iplgens - \iplgens\circ\mathcal{F}^{-1} = 0$, so finally we have $\comm{\mathcal{F}}{\iplgens}=\comm{\mathcal{F}^{-1}}{\iplgens} = 0$.
\end{proof}

The joint property of $\iplgens$ and $\mathcal{F}$ proven in above proposition is sometimes called \emph{covariance property}\cite{AlickiLidarZanardi2006}, which lets us to split the whole generator into Hamiltonian and dissipative parts, $\lgens{t} = \hsup{t}{\,\cdot\,} + \disss{t}$. It is also of basic importance for proving the existence of periodic limit cycle of $\qdms{t}{t_0}$ in \schr picture (see below).

\begin{proposition}[\textbf{properties of Lindbladian and propagator}]

We have, that
\begin{enumerate}
\item $\lgens{t}$ is periodic, $\lgens{t+T} = \lgens{t}$;
\item $\qdms{t}{t_0}$ is invariant with respect to translation by $nT$, i.e. $\qdms{t+nT}{t_0+nT} = \qdms{t}{t_0}$, $n \in \integer$.
\end{enumerate}
\end{proposition}

\begin{proof} \hfill
\vskip\baselineskip
\noindent\emph{Ad 1.} Since periodicity of $\hams(t)$ we have $-i \comm{\hams(t)}{\,\cdot\,} = -i \comm{\hams(t+T)}{\,\cdot\,}$ and it is enough to show $\disss{t} = \disss{t+T}$. Recall that $U_{t+T,t_0} = U_{t,t_0}e^{-i\hav T} = U_{t,t_0} F_{t_0}$ due to Floquet theorem. Therefore we have $\usups{t+T}{t_0} = \usups{t}{t_0} \circ \mathcal{F}$. From prop. \ref{prop:CovarianceProperty} we have $\iplgens \circ \mathcal{F} = \mathcal{F} \circ \iplgens$, which implies
\begin{align}
	\diss{t+T}{A} &= (\usups{t}{t_0} \circ (\mathcal{F} \circ \iplgens) \circ \mathcal{F}^{-1} \circ \ausups{t}{t_0})(A) = \\
	&= (\usups{t}{t_0} \circ \iplgens \circ (\mathcal{F} \circ \mathcal{F}^{-1}) \circ \ausups{t}{t_0})(A) = \nonumber\\
	&= (\usups{t}{t_0} \circ \iplgens \circ \ausups{t}{t_0})(A) = \nonumber\\
	&= \diss{t}{A}. \nonumber
\end{align}
\vskip\baselineskip
\noindent \emph{Ad 2.} Using time-ordered formula \eqref{eq:LambdaGeneral} we have
\begin{align}
	\qdms{t+nT}{t_0+nT} &= \sum_{k=0}^{\infty} \frac{1}{k!} \int\limits_{t_0 + nT}^{t+nT} dt_{1} \int\limits_{t_0 + nT}^{t+nT} dt_{2} ... \int\limits_{t_0 + nT}^{t+nT} \mathcal{T} \{ \lgens{t_1} \circ \lgens{t_2} \circ ... \circ \lgens{t_k} \} \, dt_{k},
\end{align}
which, after putting $t_{k}^{\prime} = t_{k} - nT$, yields
\begin{align}
	\qdms{t+nT}{t_0+nT} &= \sum_{k=0}^{\infty} \frac{1}{k!} \int\limits_{t_0 }^{t} dt_{1}^{\prime} \int\limits_{t_0}^{t} dt_{2}^{\prime} ... \int\limits_{t_0}^{t} \mathcal{T} \{ \lgens{t_{1}^{\prime} + nT} \circ \lgens{t_{2}^{\prime} + nT} \circ ... \circ \lgens{t_{k}^{\prime} + nT} \} \, dt_{k}^{\prime} = \\
	&= \qdms{t}{t_0} \nonumber
\end{align}
since $\lgens{t}$ is periodic as was shown above.
\end{proof}

\subsubsection{Dual forms of Lindbladian and dynamical map}

It is a well-known fact that a dual space $\traceclass^{\adj}$ is isometrically isomorphic to $\bh$ and let $i: \traceclass^{\adj} \longrightarrow \bh$ denote an isometry bijection. Then, for every functional $\varphi \in \traceclass^{\adj}$ there exists one and only one bounded operator $\Phi_{\varphi} = i(\varphi)$ such that duality pairing $\varphi (x) = (\varphi, x) = \tr{\Phi_{\varphi} x}$ holds for every $x \in \traceclass$. Let $\mathcal{A} : \traceclass \longrightarrow \traceclass$ be bounded. We define a dual map $\mathcal{A}^{\dual} : \traceclass^{\adj} \longrightarrow \traceclass^{\adj}$ to be such a map that $\varphi \circ \mathcal{A} = \mathcal{A}^{\dual} (\varphi)$ or, via duality pairing, $(\varphi, \mathcal{A}(x)) = (\mathcal{A}^{\dual}(\varphi), x)$. Since $\traceclass^{\adj} \cong \bh$, for every $\mathcal{A}^{\dual}$ there exists a corresponding map $\pi (\mathcal{A}^{\dual})$ on $\bh$ such that a duality pairing can be expressed as $\tr{\Phi_{\varphi} \, \mathcal{A}(x)} = \tr{\pi(\mathcal{A}^{\dual})(\Phi_{\varphi})x}$ where $\pi(\mathcal{A}^{\dual})$ is defined abstractly by $i \circ \mathcal{A}^{\dual} = \pi(\mathcal{A}^{\dual}) \circ i$. Here, $\pi : \enmor{\traceclass^{\adj}} \longrightarrow \enmor{\bh}$ is an operator-valued \emph{representation} of $\traceclass^{\adj}$. Since the isometry between $\mathcal{A}^{\dual}$ and $\pi(\mathcal{A}^{\dual})$, we will refer $\pi(\mathcal{A}^{\dual}) = \mathcal{A}^{\adj}$ as a \emph{map dual to} $\mathcal{A}$ and will define it via
\begin{equation}\label{eq:DualityPairing}
	\tr{y \, \mathcal{A}(x)} = \tr{\mathcal{A}^{\adj}(y) \, x}.
\end{equation}
For $\lgens{t}$ and $\qdms{t}{t_0}$ we establish their duals $\lgens{t}^{\adj}$ and  and $\qdms{t}{t_0}^{\adj}$ (in the sense of representations as above) defined on operator \cstar $\bh$, according to equalities
\begin{equation}
	\tr{A \, \lgen{t}{\rho}} = \tr{\lgens{t}^{\adj}(A) \, \rho}, \qquad \tr{A \, \qdm{t}{t_0}{\rho}} = \tr{\qdms{t}{t_0}^{\adj}(A) \, \rho}
\end{equation}
where $\rho \in \traceclass$, $A \in \bh$. Strong continuity of $\lgens{t}$ and $\qdms{t}{t_0}$  within $\traceclass$ implies that $\algens{t}$ and $\aqdms{t}{t_0}$ are continuous in pointwise weak-$\star$ topology over $\bh$ and they are tied by a relation
\begin{equation}
	\frac{d}{dt} \aqdm{t}{t_0}{A} = \left(\aqdms{t}{t_0} \circ \algens{t} \right)(A).
\end{equation}

\begin{proposition}\label{prop:DualMaps}
Dual maps $\aiplgens$, $\algens{t}$ and $\aqdms{t}{t_0}$ are given via equalities
\begin{subequations}
\begin{align}
	&\aiplgen{A} =  \sum_{\alpha\beta}\sum_{\{\omega\}}\sum_{q\in\integer} \gamma_{\alpha\beta}(\omega+q\Omega) \left(  S_{\beta}(\omega,q)^{\adj} A \, S_{\alpha}(\omega,q) - \frac{1}{2}\acomm{S_{\beta}(\omega,q)^{\adj}S_{\alpha}(\omega,q)}{A} \right),\label{eq:DualLIP} \\
	&\algen{t}{A} = i \comm{\hams(t)}{A} + (\usups{t}{t_0} \circ \aiplgens \circ \ausups{t}{t_0})(A), \label{eq:DualL} \\
	&\aqdm{t}{t_0}{A} = \left( e^{(t-t_0)\aiplgens} \circ \ausups{t}{t_0} \right)(A). \label{eq:DualQDM}
\end{align}
\end{subequations}

\end{proposition}

\begin{proof}
Computation easily comes from cyclicity and linearity of trace, therefore we will only sketch the proof. Let us provide a simplified notation such that $\gamma_{\alpha\beta}(\omega,q) = \gamma_{\alpha\beta}$ and $S_{\alpha}(\omega,q) = V_{\alpha}$ which gives
\begin{equation}
	\iplgen{\rho} = \sum \gamma_{\alpha\beta} \left(V_{\alpha}\rho V_{\beta}^{\adj} - \frac{1}{2}\acomm{V_{\beta}^{\adj}V_{\alpha}}{\rho}\right)
\end{equation}
where the sum is taken over all indices like in \eqref{eq:DualLIP}. By cyclicity it is easy to check, that $\tr{A \, V_{\alpha}\rho V_{\beta}^{\adj}} = \tr{V_{\beta}^{\adj}AV_{\alpha} \, \rho}$ and $\tr{A \, \acomm{V_{\beta}^{\adj}V_{\alpha}}{\rho}} = \tr{\acomm{V_{\beta}^{\adj}V_{\alpha}}{A} \, \rho}$ which, by linearity, yields \eqref{eq:DualLIP}.
\par
For \eqref{eq:DualL}, first verify that 
\begin{subequations}
	\begin{align}
		&\tr{-i A \comm{\hams(t)}{\rho}} = \tr{i\comm{\hams(t)}{A} \rho}, \\
		&\tr{A \, (\usups{t}{t_0} \circ \iplgens)(\sigma)} = \tr{\ausup{t}{t_0}{A}\,\iplgen{\sigma}}
	\end{align}
\end{subequations}
and put $\sigma = \ausup{t}{t_0}{\dm{t}}$. Then, after some algebra, one arrives at  
\begin{align}
	&\tr{\ausup{t}{t_0}{A} \, \iplgen{\ausup{t}{t_0}{\dm{t}}}} \\
	&= \tr{U_{t,t_0} \left[ \sum \gamma_{\alpha\beta} \left(V_{\beta}^{\adj} U_{t,t_0}^{-1} A U_{t,t_0} V_{\alpha} - \frac{1}{2}\acomm{V_{\beta}^{\adj}V_{\alpha}}{U_{t,t_0}^{-1} A U_{t,t_0}} \right) \right] U_{t,t_0}^{-1} \dm{t}} \nonumber \\
	&= \tr{(\usups{t}{t_0} \circ \aiplgens \circ \ausups{t}{t_0})(A) \, \dm{t}} \nonumber
\end{align}
which implies \eqref{eq:DualL}. By induction, one can check that $\tr{A \, \iplgens^{n}(\rho)} = \tr{(\aiplgens)^{n}(A) \, \rho}$ which naturally leads to
\begin{equation}
	\tr{A \, e^{\tau\iplgens}} = \tr{\sum_{n=0}^{\infty} \frac{\tau^{n}}{n!} (\aiplgens)^{n}(A) \, \rho} = \tr{e^{\tau\aiplgens}(A) \, \rho}
\end{equation}
which is well-defined as long as $\| \iplgens \|<\infty$. Putting $\qdms{t}{t_0} = \usups{t}{t_0} \circ e^{(t-t_0)\iplgens}$, we get
\begin{equation}
	\tr{A \, \qdm{t}{t_0}{\rho}} = \tr{\ausup{t}{t_0}{A} \, e^{(t-t_0)\iplgens}(\rho)} = \tr{\left( e^{(t-t_0)\aiplgens} \circ \ausups{t}{t_0} \right)(A) \, \rho},
\end{equation}
yielding \eqref{eq:DualQDM} and completing the proof.
\end{proof}

\subsubsection{Periodic limit cycles in \schr picture}

\begin{lemma}\label{lemma:HavSigmaCommutativity}
Averaged Hamiltonian $\hav$ commutes with stationary point $\tilde{\sigma} \in \kernel{\iplgens}$ (if exists).
\end{lemma}

\begin{proof}
Assume that $\kernel{\iplgens}$ is non-trivial and there exists unique, non-zero stationary point $\tilde{\sigma} \in \kernel{\iplgens}$. This implies that $\mathcal{F}(\iplgen{\tilde{\sigma}}) = 0$ since $\mathcal{F}$ is linear. Proposition \ref{prop:CovarianceProperty} (the covariance property) yields commutation relation $\mathcal{F} \circ \iplgens = \iplgens \circ \mathcal{F}$ implying $\iplgen{\mathcal{F}(\tilde{\sigma})} = 0$ and, by assumption of non-triviality of $\kernel{\iplgens}$, $\mathcal{F}(\tilde{\sigma}) = e^{-i \hav T} \tilde{\sigma} e^{i\hav T} = \tilde{\sigma}$. Equivalently, we have $\tilde{\sigma} e^{i\hav T} = e^{i\hav T} \tilde{\sigma}$, so $\comm{\tilde{\sigma}}{e^{i\hav T}} = \comm{\tilde{\sigma}}{\hav} = 0$.
\end{proof}

\begin{proposition}
If there exists a unique stationary point $\tilde{\sigma} \in \kernel{\iplgens}$, then its corresponding \schr picture counterpart $\sigma_{t}$ is a periodic limit cycle.
\end{proposition}

\begin{proof}
\schr picture form of $\tilde{\sigma}$ is given as $\sigma_{t} = \usup{t}{t_0}{\tilde{\sigma}}$. From lemma \ref{lemma:HavSigmaCommutativity} and from Floquet representation (see prop. \ref{prop:EvolutionFloquetTheorem}) of $U_{t,t_0}$ it follows, that
\begin{align}
	\sigma_{t+T} &= \usup{t+T}{t_0}{\tilde{\sigma}} = P_{t+T,t_0} e^{-i\hav (t-t_0)} e^{-i\hav T} \tilde{\sigma} e^{i\hav T} e^{i\hav (t-t_0)} P_{t+T,t_0}^{-1} \\ 
	&= P_{t,t_0} e^{-i\hav (t-t_0)} \tilde{\sigma} e^{i\hav (t-t_0)} P_{t,t_0}^{-1} = \usup{t}{t_0}{\tilde{\sigma}} = \sigma_{t}, \nonumber
\end{align}
implying existence of periodic limit cycle of \schr picture evolution.
\end{proof}

The periodic limit cycle is attainable after sufficiently long evolution time, i.e. if $\lim\limits_{t\to\infty} \ipdm{t} = \tilde{\sigma}$, then $\sigma_{t}$ becomes an orbit-type attractor in $\traceclass$ and all trajectories $\dm{t}$ originating in points $\dm{t_0}$ from within its basin of attraction asymptotically tend to $\sigma_{t}$. The fact of existence of such state is important from point of view of quantum thermodynamics as one often works within the steady-state regime in order to define notions of heat flows and so-called \emph{local temperatures} in context of first and second law of thermodynamics\cite{Szczygielski2013,AlickiGelbwaserKurizki2013}.

\section{Applications}
\label{sec:Applications}

\subsection{Examples}

\begin{example}[\textbf{two-level system}]\label{ex:ExampleTLS}
As a first example, let us consider a simple model of two-level system with cosinusoidal modulation, which can find some application in the field of quantum thermal machines\cite{GelbwaserErezAlickiKurizki2013}. We set a time-dependent Hamiltonian of system of interest as
\begin{equation}
	\hams(t) = \frac{1}{2} \omega_0 \sigma^{3} + \lambda \Omega \sigma^{3} \cos{\Omega t},
\end{equation}
acting on $\hilberts = \complex^2$, where $\sigma^{i} : i = 1,2,3$ denotes an appropriate Pauli matrix and $\lambda > 0$ is a \emph{small} dimensionless steering parameter. System is coupled to electromagnetic field, described by symmetrized Fock space $\hilbertr = \mathscr{F}_{+}(\mathscr{H}_{\mathrm{em.}} \otimes \complex^{2}) = \overline{\bigoplus_{N=0}^{\infty} ((\mathscr{H}_{\mathrm{em.}}\otimes\complex^{2})^{\otimes N})_{+}}$, where bar denotes a Hilbert space completion and $\mathscr{H}_{\mathrm{em.}}$ is a one-photon Hilbert space. Polarization degrees of freedom are included within $\complex^2$. $\hilbertr$ can be considered as Hilbert spaces, complete with respect to inner product $\iprod{\Psi}{\Phi} = \sum_{N=0}^{\infty} \iprod{\Psi_{N}}{\Phi_{N}}$ for $\Psi = (\Psi_{1}, \Psi_{2}, ...)$, $\Phi_{N} = (\Phi_{1}, \Phi_{2}, ...)$. The interaction Hamiltonian is given as
\begin{equation}
	\hami = (\sigma^{+} + \sigma^{-}) \otimes B(f) = \sigma^{1} \otimes B(f),
\end{equation}
where $\sigma^{\pm} = \frac{1}{2} (\sigma^{1} \pm i\sigma^{2})$ is a raising (resp. lowering) operator in $\complex^2$ and $B(f)$ is a self-adjoint linear operator from CCR (Canonical Commutation Relations) \cstar over $\hilbertr$ with $f : \real^{3} \times \{-1,1\} \longrightarrow \complex$ being sufficiently smooth and compactly supported. Putting $t_0 = 0$, we easily obtain the Floquet representation of unitary propagator $U_{t} = P_{t} e^{-i\hav t}$, where
\begin{equation}
	P_{t} = \exp{\left\{ -i\lambda \Omega \sigma^{3} \sin{\Omega t} \right\}}, \qquad \hav = \frac{1}{2} \omega_{0} \sigma^{3}
\end{equation}
which yields a simple monodromy operator $U_{T} = \exp{\left\{-\frac{1}{2} i \omega_{0} \sigma^{3}\right\}}$ and averaged Hamiltonian $\hav = \frac{1}{2}\omega_0 \sigma^3$. It is diagonalized by (Floquet) eigenvectors $\phi_{1} = (1,0)^{T}$, $\phi_{2} = (0,1)^{T}$ and the quasienergies take a particularly simple form, $\epsilon_{1} = \frac{1}{2}\omega_0$, $\epsilon_2 = -\frac{1}{2}\omega_0$. Set $\{\omega\}$ of quasifrequencies is therefore $\{ 0, \pm \omega_0\}$ with 0 being 2-fold degenerated. Interaction picture counterpart of $S = \sigma^1$ is therefore
\begin{equation}
	\tilde{S}_{t} = \mathcal{U}^{-1}_{t}(\sigma^1) = \sigma^{1} \cos{\left( \omega_0 t + 2\lambda \sin{\Omega t} \right)} - \sigma^{2} \cos{\left( \omega_0 t + 2\lambda\sin{\Omega t} \right)}.
\end{equation}
Fourier expansion of $\tilde{S}_{t}$ is infinite. However, basing on the assumption of $\lambda$ being small, we expand the above in Maclaurin series with respect to $\lambda$ and neglect $\mathcal{O}(\lambda^2)$, effectively obtaining
\begin{equation}
	\tilde{S}_{t} = \sigma^1 \left(\cos{\omega_0 t} - 2\lambda \sin{\omega_0 t}\sin{\Omega t}  \right) - \sigma^{2} \left( \sin{\omega_0 t} + 2\lambda \cos{\omega_0 t} \sin{\Omega t} \right)
\end{equation}
which can be easy expanded into Fourier series using Euler formula. Effectively, there are 6 shifted quasienergies present in this model, $\{\omega + q\Omega\} = \{ \pm(\omega_0 - \Omega), \pm \omega_0, \pm (\omega_0 + \Omega) \}$ and corresponding Fourier coefficients are
\begin{align}
	&\tilde{S}(\omega_0,1) = \left(\begin{array}{cc}0 & \lambda \\ 0 & 0\end{array} \right), \qquad \tilde{S}(\omega_0,0) = \left(\begin{array}{cc}0 & 1 \\ 0 & 0\end{array} \right), \qquad \tilde{S}(\omega_0,-1) = \left(\begin{array}{cc}0 & -\lambda \\ 0 & 0\end{array} \right), \\
	&\tilde{S}(0,0) = 0, \qquad \tilde{S}(-\omega_0, 1) = \tilde{S}(\omega_0,-1)^{T}, \qquad \tilde{S}(-\omega_0, 0) = \tilde{S}(\omega_0,0)^{T}, \nonumber \\ 
	&\tilde{S}(-\omega_0,-1) = \tilde{S}(\omega_0,1)^{T}, \nonumber
\end{align}
where matrix representations of all operators are explicitly written in Floquet basis $\{\phi_{1},\,\phi_{2}\}$. Let $\ipdm{t} = [(\tilde{\rho}_{t})_{ij}]_{i,j=1}^{2}$, $(\tilde{\rho}_{t})_{ij} = \iprod{\phi_{i}}{\ipdm{t}\phi_{j}}$. Semigroup generator $\iplgens$ takes a form
\begin{align}
	&\frac{d}{dt} \ipdm{t} = \iplgen{\ipdm{t}} = \left( \begin{array}{cc} -a (\tilde{\rho}_{t})_{11} + b (\tilde{\rho}_{t})_{22} & -c (\tilde{\rho}_{t})_{12} \\ -c (\tilde{\rho}_{t})_{21} & a (\tilde{\rho}_{t})_{11} - b (\tilde{\rho}_{t})_{22} \end{array} \right), \\
	&a = \gamma(-\omega_0) + \lambda^{2} (\gamma(-\omega_0-\Omega)+\gamma(-\omega_0+\Omega)), \\
	&b = \gamma(\omega_0) + \lambda^{2}(\gamma(\omega_0-\Omega)+\gamma(\omega_0+\Omega)), \nonumber \\
	&c = \frac{1}{2} (\gamma(-\omega_0)+\gamma(\omega_0)+\lambda^{2}(\gamma(-\omega_0-\Omega)+\gamma(-\omega_0+\Omega)+\gamma(\omega_0-\Omega)+\gamma(\omega_0+\Omega))).\nonumber
\end{align}
where $\gamma(x) = Ax^{3}/ (1-e^{-\beta_{\mathrm{e}}x})$ is the spectral density function of environment (see e.g. ref. \onlinecite{Szczygielski2013}) and $\beta_{\mathrm{e}}$ is the inverse temperature of electromagnetic field, $\beta_{\mathrm{e}} = 1/T_{\mathrm{e}}$. We note that in case of equilibrium environment, the famous KMS (Kubo--Martin--Schwinger) condition\cite{AlickiLendi} $\gamma(-x) = e^{-\beta_{\mathrm{e}}x} \gamma (x)$ will hold and formulas can be further simplified. Notice $\tr{\iplgen{\ipdm{t}}} = 0$, as desired for trace-preserving. Resulting interaction picture semigroup $e^{t \iplgens}$, valid for small $\lambda$ regime, is then given by
\begin{align}
	&e^{t \iplgens} (\tilde{\rho}_0) = \left( \begin{array}{cc} e^{-(a+b)t}\left[f(t) (\tilde{\rho}_{0})_{11} + b g(t) (\tilde{\rho}_{0})_{22}\right] & e^{-ct} (\tilde{\rho}_{0})_{12} \\ e^{-ct} (\tilde{\rho}_{0})_{21} & e^{-(a+b)t} \left[a g(t) (\tilde{\rho}_{0})_{11} + h(t) (\tilde{\rho}_{0})_{22}\right] \end{array} \right), \\
	&f(t) = \frac{a+b e^{(a+b)t}}{a+b}, \qquad g(t) = \frac{e^{(a+b)t}-1}{a+b}, \qquad h(t) = \frac{b+a e^{(a+b)t}}{a+b},
\end{align}
where $\tilde{\rho}_0$ is an initial density operator of trace one. Off-diagonal terms of $\ipdm{t}$ vanish after long times and density operator decoheres into diagonal stationary state $\tilde{\sigma} = (a+b)^{-1} \mathrm{diag}\{ b,a \}$. To obtain full dynamical map $\Lambda_{t}$, simply apply $\mathcal{U}_{t}$ to obtain
\begin{equation}
	\Lambda_{t}(\dm{0}) = \left( \begin{array}{cc}
	e^{-(a+b)t} \left[ f(t)(\tilde{\rho}_{0})_{11} + b g(t) (\tilde{\rho}_{0})_{22} \right] & e^{-t(c+\omega_0 i)} e^{-2\lambda i \sin{\Omega t}} (\tilde{\rho}_{0})_{12} \\
	e^{-t(c-\omega_0 i)} e^{2\lambda i \sin{\Omega t}} (\tilde{\rho}_{0})_{21} & e^{-(a+b)t} \left[ a g(t) (\tilde{\rho}_{0})_{11} + h(t) (\tilde{\rho}_{0})_{22} \right] 
	\end{array}\right).
\end{equation}
In this example, $U_{t}$ is diagonal and commutes with $\tilde{\sigma}$; therefore $\Lambda_{t}(\dm{0}) \to \sigma_{t} = \tilde{\sigma}$ for large $t$. This is a special case of trivialized orbit as, asymptotically, the trajectory approaches to a fixed point.
\end{example}

\begin{example}[\textbf{harmonic oscillator}]\label{ex:HO}
Now, consider a one-mode harmonic oscillator with characteristic frequency $\omega$, driven by external monochromatic laser beam of frequency $\Omega$, weakly interacting with electromagnetic field. The space $\hilberts$ of system of interest is taken as bosonic, symmetrized Fock space over one-state Hilbert space $\complex$, defined as a completion $\hilberts = \mathscr{F}_{+} (\complex) = \overline{\bigoplus_{N=0}^{\infty} (\complex ^{\otimes N})_{+}}$ and environment space $\hilbertr$ is as in example \ref{ex:ExampleTLS}. From mathematical point of view, this example features a more general framework, as we explicitly introduce unbounded bosonic operators. However, such models are very common in realm of open systems theory and appropriate Lindblad-like formulas for semigroup generator still make sense and generate a valid dynamical maps despite unboundedness. It remains true also in the Floquet-related approach, therefore we provide this illustrative computation. Second-quantized Hamiltonian of system of interest is
\begin{equation}
	\hams(t) = \omega \, a^{\adj} a + g (e^{i\Omega t} a + e^{-i\Omega t} a^{\adj}),
\end{equation}
where $a$ and $a^{\adj}$ are standard creation and annihilation operators on $\mathscr{F}_{+}(\complex)$ satisfying $\comm{a}{a^{\adj}} = \id{\hilberts}$ and $g$ is a constant characterizing interaction with laser beam. Laser light is treated quasi-classically and interaction, given under rotating wave approximation (RWA) is treated as a time-dependent perturbation. System is again coupled to environment by $\hami = (a + a^{\adj}) \otimes B(f)$. The unitary propagator takes a general, time-ordered form \eqref{eq:UnitaryPropagatorTimeOrdered} which may be very hard to compute explicitly. However, one can easily check that Floquet representation $U_t = P_{t} e^{-i\hav t}$ of propagator ($t_0 = 0$) is given by
\begin{align}
	P_{t} = e^{-it \Omega \, a^{\adj}a }, \qquad \hav = \Delta \, a^{\adj} a + g(a + a^{\adj})
\end{align}
with $\Delta = \omega - \Omega$ being the detuning parameter. To verify that such particular choice indeed constitutes the Floquet representation, simply differentiate it with respect to $t$ and check that is satisfies the same differential equation as $U_{t}$, namely $\frac{d}{dt}(P_{t}e^{-i\hav t}) = -i \hams(t) \, P_{t}e^{-i\hav t}$. The corresponding map $\mathcal{U}_{t}^{-1}$ is still $\mathcal{U}_{t}^{-1}(A) = U_{t}^{-1} A U_{t}$. Monodromy operator $U_{T}$ is
\begin{equation}
	U_{T} = e^{-2\pi i \, a^{\adj}a} e^{-iT\hav} = e^{-iT\hav}
\end{equation}
as $e^{-2\pi i \, a^{\adj}a} = \id{\hilberts}$. Now one can apply a following unitary transformation of $a$ and $a^{\adj}$,
\begin{equation}
	c = a - \alpha = W_{\alpha} a W_{\alpha}^{\adj}, \qquad c^{\adj} = a^{\adj} - \alpha = W_{\alpha} a^{\adj} W_{\alpha}^{\adj},
\end{equation}
where $\alpha = -g/\Delta$ and $W_{\alpha} = \exp{\{\alpha a^{\adj} - \overline{\alpha}a\}}$ is a unitary Weyl displacement operator. New operators are subject to the same CCR-algebraic relation $\comm{c}{c^{\adj}} = \id{\mathscr{F}_{+}(\complex)}$ and therefore $\hav = \Delta (c^{\adj}c - \alpha^2 )$ and
\begin{equation}
	U_{t} = e^{-it\Omega \, a^{\adj}a} e^{-it\Delta( c^{\adj}c - \alpha^{2})}.
\end{equation}
$\hav$ has pure-point spectrum and is diagonalized by (Floquet) eigenvectors $\phi_{n} = W_{\alpha} \psi_{n}$, $n \geqslant 0$, where $\psi_{n}$ is an eigenvector of excitation number operator $a^{\adj}a$ such that $a^{\adj}a \, \psi_{n} = n\psi_{n}$. Therefore $c^{\adj}c \, \phi_{n} = n \phi_{n}$ and quasienergies are of a form $\epsilon_{n} = \Delta (n - \alpha^{2})$. Correspondingly, spectrum of $F$ is of a form $e^{-i\epsilon_{n} T} = e^{ -iT\Delta(n-\alpha^{2}) }$. Set of quasifrequencies is then
\begin{equation}
	\{\omega\} = \{ \epsilon_{n} - \epsilon_{m} : n,m\in \nat \} = \{ \Delta (n-m) \} = \{ \Delta k : k \in \integer\}.
\end{equation}
System's part of interaction Hamiltonian is $S = a+a^{\adj}$. Applying $\mathcal{U}_{t}^{-1}$, one obtains after some algebra a following expression for interaction picture $\tilde{S}_{t}$,
\begin{equation}\label{eq:StFourierTerms}
	\tilde{S}_{t} = e^{-it(\Delta + \Omega)}c + e^{it(\Delta + \Omega)} c^{\adj} + \alpha e^{-it\Omega} \id{\mathscr{F}_{+}(\complex)} + \alpha e^{it\Omega} \id{\mathscr{F}_{+}(\complex)},
\end{equation}
from which we see that there are effectively only 3 quasifrequencies present, $\{0, \pm \Delta\}$ and 4 shifted quasifrequencies, $\{\pm \Omega, \pm (\Delta + \Omega)\}$. Therefore a (diagonalized) Floquet-Lindblad generator $\iplgens$ in interaction picture becomes
\begin{equation}\label{eq:HOscGenerator}
	\iplgen{\ipdm{t}} = \gamma_1 \left(  c \, \ipdm{t} \, c^{\adj} - \frac{1}{2}\acomm{c^{\adj}c}{\ipdm{t}} \right) + \gamma_2 \left( c^{\adj}\,\ipdm{t}\,c - \frac{1}{2}\acomm{c\,c^{\adj}}{\ipdm{t}} \right)
\end{equation}
where $\gamma_{1,2} = \gamma(\pm (\Delta + \Omega))$ and $\gamma(x) = Ax^{3}/ (1-e^{-\beta_{\mathrm{R}}x})$ is the spectral density function of environment (see e.g. ref. \onlinecite{Szczygielski2013}) and $\beta_{\mathrm{R}}$ is the inverse temperature of electromagnetic field. Note that Fourier terms in \eqref{eq:StFourierTerms} corresponding to $\pm\Omega$ are deprecated since the are proportional to identity and therefore provide no contribution to $\iplgens$.
\par
This type of model is exactly solvable (see ref. \onlinecite{AlickiLendi} and references therein) via expression for \emph{dual} $\tilde{\Lambda}_{t}^{\adj}$, understood in a sense of proposition \ref{prop:DualMaps}. One can check, that $\tilde{\Lambda}_{t}^{\adj}$, defined by its action on Weyl operator $W(z)$, $z \in \complex$ as 
\begin{equation}
	\tilde{\Lambda}_{t}^{\adj}(W(z)) = \exp{\left\{ -\frac{|z|^{2}}{2} \frac{\gamma_{1}}{\gamma_{1}-\gamma_{2}} \left( 1-e^{-(\gamma_{1}-\gamma_{2})t} \right) \right\}} W\left(\frac{2i\overline{z}}{\sqrt{2}}\exp{\left\{ -\frac{1}{2}(\gamma_{1}-\gamma_{2})t \right\}}\right),
\end{equation}
indeed satisfies the equation
\begin{equation}
	\frac{d}{dt} \tilde{\Lambda}_{t}^{\adj}(W(z)) = (\aiplgens \circ \tilde{\Lambda}_{t}^{\adj})(W(z))
\end{equation}
where $\aiplgens$ is a generator dual to \eqref{eq:HOscGenerator}. It implies that $\tilde{\Lambda}_{t}^{\adj}$ is a dual to interaction picture semigroup $\tilde{\Lambda}_{t}$. This is a formal -- however exact -- solution; other formulations of dynamical semigroup are also possible\cite{AlickiLendi}. As long as $\gamma_{1} > \gamma_{2}$ (which is fulfilled via Kubo -- Martin -- Schwinger condition in case of thermally-equilibrium reservoir\cite{AlickiLendi}), $\tilde{\Lambda}_{t}^{\adj}(W(z))$ tends $W_{\infty}(z) = \exp{\{-|z|^{2} \gamma_{1} / 2(\gamma_{1}-\gamma_{2}) \}} \cdot \id{\hilbertr}$ in weak-$\star$ topology as $t \to \infty$. This implies, that $\tilde{\Lambda}_{t}(\dm{0})$ tends to unique stationary point which is identified with thermal state
\begin{equation}
	\tilde{\sigma} = \frac{e^{-\beta \Delta c^{\adj} c}}{1-e^{-\beta \Delta}} = \frac{e^{-\beta (\hav + \Delta\alpha^{2})}}{1-e^{-\beta \Delta}},
\end{equation}
where $\beta = \Delta^{-1} \ln{(\gamma_{1}/\gamma_{2})}$. In \schr picture, this thermal state yields a periodic orbit 
\begin{equation}
	\sigma_{t} = U_{t} \sigma U_{t}^{-1} = \frac{1}{1-e^{-\beta \Delta}} e^{-it\Omega \, a^{\adj}a} e^{-\beta(\hav + \Delta\alpha^{2})} e^{it\Omega \, a^{\adj}a}.
\end{equation}
Notice that when no external driving is applied ($g=0$, $\Omega = 0$, $\Delta = \omega$, $\alpha = 0$) the averaged Hamiltonian $\hav$ gets replaced by $\omega \, a^{\adj}a$ and one instead gets a stationary Gibbs state $\sigma_{\beta} = e^{-\beta \omega \, a^{\adj}a} / (1-e^{-\beta \omega})$, $\beta = \omega^{-1} \ln{(\gamma(\omega)/\gamma(-\omega))}$.
\end{example}

\begin{example}\label{ex:TwoBaths}
In ref. \onlinecite{Szczygielski2013} a model of two-level system described by time-dependent Hamiltonian
\begin{equation}
	\hams(t) = \frac{1}{2} \omega_{0} \sigma^{3} + g (e^{i\Omega t} \sigma^{-} + e^{-i\Omega t} \sigma^{+})
\end{equation}
and coupled to electromagnetic field through
\begin{equation}
	H_{\mathrm{int., e}} = \sigma^{1} \otimes B(f)
\end{equation}
was analyzed. Appropriate interaction picture dynamics was developed and various thermodynamical features of this model were addressed. In particular, two thermodynamical regimes were studied.
\par
Firstly, it was explicitly assumed that the field remains in 0 temperature and in a vacuum state. In this case, the spectral density function is modified such that
\begin{equation}
	\gamma_{\mathrm{e}}(x) = A x^{3} \chi_{[0,\infty)}(x)
\end{equation}
(negative frequencies are cut off). The time-dependent part of $\hams(t)$ describes (under rotating wave approximation) the action of monochromatic laser beam of frequency $\Omega$ and it was assumed that the detuning parameter $\Delta = \omega_0 - \Omega$ could be arbitrary. Deriving appropriate Markovian master equation in interaction picture, the general phenomenon of nonresonant fluorescence of such system was described and a formula for fluorescence power spectrum was obtained.
\par
Secondly, the field was put in equilibrium state of finite temperature $T_{\mathrm{e}}$, described by spectral density of form
\begin{equation}
	\gamma_{\mathrm{e}}(x) = \frac{Ax^3}{1-e^{\beta_{\mathrm{e}}x}},
\end{equation}
$T_{\mathrm{e}} = 1/\beta_{\mathrm{e}}$, and two-level system was coupled to additional, so-called \emph{dephasing bath} of some temperature $T_{\mathrm{d}} = 1/\beta_{\mathrm{d}}$ and spectral density $\gamma_{\mathrm{d}}(x)$. In principle, the exact form of $\gamma_{\mathrm{d}}(x)$ may be unknown, however its values at certain points $\{ \omega + q\Omega \}$ can be sometimes determined experimentally, depending on exact realization of a model. The coupling was implemented through interaction Hamiltonian of a form
\begin{equation}
	H_{\mathrm{int., d}} = \sigma^{3} \otimes F
\end{equation}
with $F = F^{\adj}$ acting on Hilbert space of states of dephasing bath. Without external driving, this Hamiltonian is responsible for pure decoherence effects (hence the name of a bath) only, leaving diagonal terms of $\dm{t}$ (\emph{populations}) unchanged. It was shown, that resulting dynamical semigroup allows to interpret the whole system as a heat pump, which generates a heat flow between baths and direction of this flow depends on $\mathrm{sgn}\, (\Delta)$. 
\end{example}

\subsection{Note on some auxiliary results}

The approach outlined in this paper was recently used several times, mainly in context of various models of quantum thermal machines. In most scenarios, a two-level system coupled to two heat baths at different temperatures and driven by external coherent light source was used as a substantial building block of simple (but effective) microscopic machine, pumping heat from one bath to another. Such system is easily described by time-dependent Hamiltonian $H(t) = \frac{1}{2} \omega(t) \sigma^{3}$, where $\sigma^{3} = \mathrm{diag}\{1,-1\}$ is a Pauli matrix and $\omega (t)$ is periodically modulated, or by more general form $H(t) = \frac{1}{2} \omega_0 \sigma^{3} + V(t)$, with $\omega_0$ being an unperturbed characteristic frequency of two-level system and $V(t)$ standing for periodic perturbation (not necessarily commuting with $\sigma^{3}$, as in example \ref{ex:TwoBaths}).
\par
Such periodically driven quantum heat machines were shown to be \emph{universal} in a sense that the are able to act as quantum engine or as quantum refrigerator, pumping heat from cold bath to hot one or vice-versa, depending on light modulation\cite{LevyAlickiKosloff2006,AlickiGelbwaserKurizki2013}. Refrigerator regime gained even more attention, also in the context of much more fundamental -- and more challenging -- issues such as unattainability of absolute zero temperature, being a consequence of Nernst' formulation of third law of thermodynamics\cite{AlickiKolarGelbwaserKurizki2012}.

\section{Concluding remarks and open problems}

It was shown that there exists a mathematically rigorous and self-consistent description of open systems governed by periodic Hamiltonians in terms of composite, trace-preserving dynamical maps. The result seems plausible at least because of the fact, that resulting time-dependent Lindbladian $\lgens{t}$ emerges from underlying, time-independent generator $\iplgens$ of \emph{ordinary} dynamical semigroup which is handled by well-known methods. Moreover, directly from Floquet theory it follows that such representation of $\qdms{t}{t_0}$ exists if only monodromy operator is normal (which is guaranteed by self-adjointness of Hamiltonian) and has discrete spectrum. From purely computational point of view, the approach to general completely positive dynamics based on Floquet theory seems to be a robust and powerful tool, allowing to find an exact form of Lindbladian. This field of research seems very promising, as quantum systems featuring a periodic modulation emerge quite naturally in quantum optics or nanotechnology. This idea of modulated quantum engine allows to rethink many statements, regarding our understanding of thermodynamics on quantum level. Naturally, spectrum of possible practical implementations of such microscopic devices seems to be rich and includes various incarnations of \emph{quantum machine} idea such as some externally modulated nanosystems, quantum dots, optically active atoms etc. as well as more complicated ones, perhaps even biological.
\par
Naturally, there are still some unanswered questions remaining. One of them is related to the more general problem of existence and uniqueness of stationary point, which was here simply assumed to be unique. There is some literature present, starting from paper by Frigerio \cite{Frigerio1978}, which addresses a general problem of existence of faithful stationary points; however, this topic still needs more attention, especially in case of infinite-dimensional spaces. An interesting direction for eventual progress is definitely the much more demanding non-periodic regime where, for example, $H(t) = H_{1}(t) + H_{2}(t)$ with $H_{1,2}(t)$ being periodic with non-commensurate periods $T_{1,2}$. This framework lays in the domain of \emph{multi-mode Floquet theory} and is difficult already in the realm of unitary dynamics; however there are some significant advances in this field, at least for bichromatic case\cite{ChuTelnov2004,ErnstSamosonMeier2005}. 

\section*{Acknowledgements}

Author is very grateful to Professor Robert Alicki for fruitful discussions and valuable suggestions and to David Gelbwaser-Klimovsky for comments. Support by University of Gdansk (via grant No. 538-5400-B166-13) and by the Foundation for Polish Science TEAM project (cofinanced by the EU European Regional Development Fund) is greatly acknowledged.
\vskip\baselineskip
This paper contains, in a large part, results of work stimulated and influenced by conference \emph{Mathematical Horizons for Quantum Physics 2}, organized jointly by Institute for Mathematical Sciences and Centre for Quantum Technologies of National University of Singapore in 2013. Author acknowledges the support received from organizers during this event.


\appendix

\section{Derivation of semigroup generator in the interaction picture}\label{app:MME}

Given Markovian master equation \eqref{eq:PrototypeMME},
\begin{equation}
	\frac{d}{dt} \ipdm{t} = -\int\limits_{t_0}^{\infty} \ptr{\hilbertr}{\comm{\iphami{t}}{\comm{\iphami{t-t^{\prime}}}{\ipdm{t}\otimes\omega}}} \, dt^{\prime},
\end{equation}
and applying equality $\comm{A}{BC}-\comm{A}{CB} = ABC-BCA + \hc$, we compute the double commutator to obtain
\begin{equation}
	\frac{d}{dt} \ipdm{t} = \int\limits_{t_0}^{t} (A_{1}-A_{2}) \, dt^{\prime} + \hc,
\end{equation}
where
\begin{subequations}
\begin{align}
&A_{1} = \ptr{\hilbertr}{\iphami{t-t^{\prime}} (\ipdm{t}\otimes\dmr)\iphami{t}^{\adj}}, \\
&A_{2} = \ptr{\hilbertr}{\iphami{t}^{\adj} \iphami{t-t^{\prime}} (\ipdm{t}\otimes\dmr)},
\end{align}
\end{subequations}
and $\iphami{t}$ was intentionally replaced by $\iphami{t}^{\adj}$. Expanding $\iphami{t}$ and $\iphami{t-t^{\prime}}$ according to \eqref{eq:IPHIexpansion} and \eqref{eq:IPSexpansion}, we get
\begin{subequations}
	\begin{align}
		A_{1} &= \lambda^{2}\sum_{\alpha\beta} r_{\beta\alpha}(t,t^{\prime}) \ips{\alpha}{t-t^{\prime}} \ipdm{t} \ips{\beta}{t}^{\adj} \\
		&= \sum_{\alpha\beta}\sum_{qq^{\prime}}\sum_{\omega\omega^{\prime}}e^{-i(\omega+q\Omega)t^{\prime}} r_{\beta\alpha}(t,t^{\prime}) e^{i(\omega-\omega^{\prime}+(q-q^{\prime})\Omega)(t-t_0)} S_{\alpha}(\omega,q) \ipdm{t} S_{\beta}(\omega^{\prime},q^{\prime})^{\adj}, \nonumber
	\end{align}
	\begin{align}
		A_{2} &= \lambda^{2}\sum_{\alpha\beta} r_{\beta\alpha}(t,t^{\prime}) \ips{\beta}{t}^{\adj} \ips{\alpha}{t-t^{\prime}}\ipdm{t} \\
		&= \sum_{\alpha\beta}\sum_{qq^{\prime}}\sum_{\omega\omega^{\prime}} e^{-i(\omega+q\Omega)t^{\prime}} r_{\beta\alpha}(t,t^{\prime}) e^{i(\omega-\omega^{\prime}+(q-q^{\prime})\Omega)(t-t_0)} S_{\beta}(\omega^{\prime},q^{\prime})^{\adj} S_{\alpha}(\omega,q)  \ipdm{t}, \nonumber
	\end{align}
\end{subequations}
where the \emph{reservoir autocorrelation function} was introduced,
\begin{equation}\label{eq:AutocorrelationFunction}
	r_{\alpha\beta}(t,s) = \tr{\dmr\ipr{\alpha}{t}\ipr{\beta}{t-s}}.
\end{equation}

\begin{proposition}\label{prop:AutocorrFunHomogenous}
If $\comm{\dmr}{\hamr} = 0$, i.e. $\dmr$ expresses a constant density operator of the environment, then the autocorrelation functions $r_{\alpha\beta}$ are homogenous in time, i.e. $r_{\alpha\beta}(t+\tau,s) = r_{\alpha\beta}(t,s)$.
\end{proposition}

\begin{proof}
Let $U_{t,t_0} = e^{-i\hamr(t-t_0)}$ be a two-parameter unitary group. From definition \eqref{eq:AutocorrelationFunction} and \eqref{eq:IPRdefinition} it follows, that $\ipr{\alpha}{t} = U_{t,t_0}^{-1} R_{\alpha} U_{t,t_0}$, which implies
\begin{align}
	r_{\alpha\beta}(t,s) &= \ptr{\hilbertr}{\dmr U_{t,t_0}^{-1} R_{\alpha} U_{t,t_0}U_{t-s,t_0}^{-1} R_{\beta} U_{t-s,t_0}} \\
	&= \ptr{\hilbertr}{\dmr U_{t-s,t_0} U_{t,t_0}^{-1} R_{\alpha} U_{t,t_0}U_{t-s,t_0}^{-1} R_{\beta}} \nonumber \\
	&= \ptr{\hilbertr}{\dmr U_{t-s+\tau,t_0} U_{t+\tau,t_0}^{-1} R_{\alpha} U_{t+\tau,t_0}U_{t-s+\tau,t_0}^{-1} R_{\beta}} \nonumber \\
	&= \ptr{\hilbertr}{\dmr U_{t+\tau,t_0}^{-1} R_{\alpha} U_{t+\tau,t_0}U_{t-s+\tau,t_0}^{-1} R_{\beta} U_{t-s+\tau,t_0}} \nonumber \\
	&= \ptr{\hilbertr}{\dmr \ipr{\alpha}{t+\tau} \ipr{\beta}{t-s+\tau}} \nonumber \\
	&= r_{\alpha\beta}(t+\tau,s) \nonumber 
\end{align}
which is implied by cyclic property of trace, $\tr{ABC} = \tr{CAB} = \tr{BCA}$, properties of $U_{t,t_0}$ as a two-parameter unitary group, $U_{t_{1},t_{2}} U_{t_{3},t_{4}} = U_{t_{1}+t_{3},t_{2}+t_{4}}$, $U_{t,t_0}^{-1} = U_{t_0,t}$ and the assumption of $\dmr$ being constant, $\comm{\dmr}{U_{t,t_0}} = 0$.
\end{proof}
\noindent Markovian master equation now takes a form
\begin{subequations}
\begin{align}\label{eq:MMEnoSecular}
	\frac{d}{dt} \ipdm{t} &= \lambda^{2}\sum_{\alpha\beta}\sum_{\omega\omega^{\prime}}\sum_{qq^{\prime}} \Gamma_{\beta\alpha}(\omega+q\Omega) e^{i(\omega-\omega^{\prime} + (q-q^{\prime})\Omega)(t-t_0)} Z(\alpha,\beta,\omega,\omega^{\prime},q,q^{\prime})(\ipdm{t}) \\
	&+ \hc, \nonumber \\
	&Z(\alpha,\beta,\omega,\omega^{\prime},q,q^{\prime})(\ipdm{t}) = S_{\alpha}(\omega,q) \ipdm{t} S_{\beta}(\omega^{\prime},q^{\prime})^{\adj} - S_{\beta}(\omega^{\prime},q^{\prime})^{\adj} S_{\alpha}(\omega,q)\ipdm{t},
\end{align}
\end{subequations}
and $\Gamma_{\alpha\beta}$ are given by one-sided Fourier transforms,
\begin{equation}
	\Gamma_{\alpha\beta}(x) = \int\limits_{t_0}^{\infty} e^{-ixt^{\prime}} \tr{\dmr \ipr{\alpha}{t}\ipr{\beta}{t-t^{\prime}}} \, dt^{\prime}.
\end{equation}
The next step one wants to make in order to simplify this expression is usually referred as \emph{secular approximation}, which states, roughly, that \emph{nonsecular terms} in \eqref{eq:MMEnoSecular}, i.e. those for which $\omega \neq \omega^{\prime}$ and $q \neq q^{\prime}$, may be neglected since the oscillate very rapidly\cite{AlickiLidarZanardi2006,BreuerPetruccione2002,AlickiLendi}. This can be justified, since one is averaging over sufficiently long times, $t - t_0 \gg \max_{\omega\neq\omega^{\prime},m\in\integer}\{|\omega-\omega^{\prime}+m\Omega|^{-1}\}$ and only slowly varying terms remain\cite{AlickiLidarZanardi2006}. We estimate the typical intrinsic evolution time $\tau_{\mathrm{S}}$ of system $\mathrm{S}$ to be comparable with $\max_{\omega\neq\omega^{\prime},m\in\integer}\{|\omega-\omega^{\prime}+m\Omega|^{-1}\}$ and much larger than both relaxation time $\tau_{\mathrm{R}}$ of $\mathrm{S}$, e.g. a time during which system's state $\dm{t}$ changes sufficiently, and correlation decay time $\tau_{\mathrm{R}}$ of reservoir, $\tau_{\mathrm{S}} \gg \tau_{\mathrm{R}} \gg \tau_{\mathrm{B}}$. We are left with
\begin{equation}\label{eq:MMEnotordered}
	\frac{d}{dt} \ipdm{t} = \lambda^{2}\sum_{\alpha\beta}\sum_{\{\omega\}}\sum_{q\in\integer} \Gamma_{\beta\alpha}(\omega+q\Omega) \left(S_{\alpha}(\omega,q) \ipdm{t} S_{\beta}(\omega,q)^{\adj} - S_{\beta}(\omega,q)^{\adj} S_{\alpha}(\omega,q)\ipdm{t} \right) + \hc
\end{equation}

\begin{proposition}\label{prop:GammaMatrixResolution}
There exists a function $\Gamma_{\alpha\beta}^{\prime}$ such that $\Gamma_{\alpha\beta}$ may be expressed as
\begin{equation}\label{eq:GammaResolution}
	\Gamma_{\alpha\beta}(x) = \frac{1}{2}g_{\alpha\beta}(x) + i \Delta_{\alpha\beta}(x)
\end{equation}
where:
\begin{enumerate}
	\item $g_{\alpha\beta}(x) = \Gamma_{\alpha\beta}(x) + \Gamma^{\prime}_{\alpha\beta}(x)$ and $\Delta_{\alpha\beta}(x) = \frac{1}{2i} \left(\Gamma_{\alpha\beta}(x) - \Gamma^{\prime}_{\alpha\beta}(x)\right)$, \label{property:property1}
	\item $\Delta_{\alpha\beta}(x) = \sigma_{\alpha\beta}(x) + \pi_{\alpha\beta}(x)$ where $[\sigma_{\alpha\beta}(x)]$ and $[\pi_{\alpha\beta}(x)]$ is respectively hermitian and antihermitian matrix,\label{property:property2}
	\item $[g_{\alpha\beta}(x)+2i\pi_{\alpha\beta}(x)]$ is hermitian and positive.  \label{property:property3}
\end{enumerate}
\end{proposition}

\begin{proof}
Let us assume a particular form of $\Gamma_{\alpha\beta}^{\prime}$,
\begin{equation}\label{eq:GammaDefinition}
	\Gamma_{\alpha\beta}^{\prime} (x) = \int\limits_{-\infty}^{t_0} e^{-ixt^{\prime}} \tr{\dmr \ipr{\alpha}{t}\ipr{\beta}{t-t^{\prime}}} \, dt^{\prime}.
\end{equation}
We will show that such a function satisfies all needed requirements. Property \ref{property:property1} is trivial -- just substitute proposed definitions of $g_{\alpha\beta}$ and $\Delta_{\alpha\beta}$ to \eqref{eq:GammaResolution}. For property \ref{property:property2}, let us split $\Gamma_{\alpha\beta}^{\prime}(x)$ into $\Gamma_{\alpha\beta}^{\prime}(x) = \tilde{\Gamma}_{\alpha\beta}(x) + A_{\alpha\beta}(x)$ such that
\begin{subequations}
\begin{align}
	&\tilde{\Gamma}_{\alpha\beta}(x) = \int\limits_{-\infty}^{-t_0} e^{-ixt^{\prime}} \tr{\dmr \ipr{\alpha}{t}\ipr{\beta}{t-t^{\prime}}} \, dt^{\prime}, \\ 
	&A_{\alpha\beta}(x) = \int\limits_{-t_0}^{t_0} e^{-ixt^{\prime}} \tr{\dmr \ipr{\alpha}{t}\ipr{\beta}{t-t^{\prime}}} \, dt^{\prime}.\label{eq:ADefinition}
\end{align}
\end{subequations}
Using cyclicity of trace and property $\overline{\tr{A}} = \tr{A^{\adj}}$ one has
\begin{align}
	\overline{\Gamma_{\alpha\beta}(x)} &= \int\limits_{t_0}^{\infty} e^{ixt^{\prime}} \tr{\dmr \ipr{\beta}{t-t^{\prime}}\ipr{\alpha}{t}} \, dt^{\prime}.
\end{align}
Putting $u = -t^{\prime}$, $du = -dt^{\prime}$ and applying prop. \ref{prop:AutocorrFunHomogenous},
\begin{align}
	\overline{\Gamma_{\alpha\beta}(x)} &= \int\limits_{-\infty}^{-t_0} e^{-ixu} \, \tr{\dmr \ipr{\beta}{t+u}\ipr{\alpha}{t}} \, du \\
	&= \int\limits_{-\infty}^{-t_0} e^{-ixu} \, \tr{\dmr \ipr{\beta}{t}\ipr{\alpha}{t-u}} \, du = \tilde{\Gamma}_{\beta\alpha}(x), \nonumber
\end{align}
so one obtains $\Gamma^{\prime}_{\alpha\beta}(x) = \overline{\Gamma_{\beta\alpha}(x)} + A_{\alpha\beta}(x)$. Calculating similarly one shows that $[A_{\alpha\beta}(x)]$ is hermitian. Take
\begin{equation}
	\sigma_{\alpha\beta}(x) = \frac{1}{2i}\left(\Gamma_{\alpha\beta}(x)-\overline{\Gamma}_{\beta\alpha}(x)\right), \qquad \pi_{\alpha\beta}(x) = -\frac{1}{2i}A_{\alpha\beta}(x)
\end{equation}
and conclude, by elementary calculations, that $\overline{\sigma_{\alpha\beta}(x)} = \sigma_{\beta\alpha}(x)$ and $\overline{\pi_{\alpha\beta}(x)} = -\pi_{\beta\alpha}(x)$, i.e. we have defined a hermitian and anit-hermitian matrix. Of course $\Delta_{\alpha\beta}(x) = \sigma_{\alpha\beta}(x) + \pi_{\alpha\beta}(x)$ as intended and property \ref{property:property2} is proved. For property \ref{property:property3}, let us first note, that $g_{\alpha\beta}(x) + 2i\pi_{\alpha\beta}(x) = g_{\alpha\beta}(x) - A_{\alpha\beta}(x)$. Showing hermiticity is trivial, since $[A_{\alpha\beta}(x)]$ is easily shown to be hermitian and $g_{\alpha\beta}(x)$ is hermitian as well,
\begin{equation}
	g_{\alpha\beta}(x) = \Gamma_{\alpha\beta}(x) + \overline{\Gamma_{\beta\alpha}(x)} + A_{\alpha\beta}(x), \qquad \overline{g_{\alpha\beta}(x)} = \Gamma_{\beta\alpha}(x) + \overline{\Gamma_{\alpha\beta}(x)} + A_{\beta\alpha}(x) = g_{\beta\alpha}(x).
\end{equation}
Function $g_{\alpha\beta}(x)$ is, by \eqref{eq:GammaDefinition} and \eqref{eq:ADefinition}, actually equal to Fourier transform of time-dependent correlation function of reservoir,
\begin{equation}
	g_{\alpha\beta}(x) = \int\limits_{-\infty}^{\infty} e^{-ixt^{\prime}} \tr{\dmr \ipr{\alpha}{t}\ipr{\beta}{t-t^{\prime}}} \, dt^{\prime}
\end{equation}
and
\begin{align}
	g_{\alpha\beta}(x) - A_{\alpha\beta}(x) = \int\limits_{-\infty}^{\infty} e^{-ixt^{\prime}} \tr{\dmr \ipr{\alpha}{t}\ipr{\beta}{t-t^{\prime}}}\left(1-\chi_{[-t_0,t_0]}(t^{\prime})\right) \, dt^{\prime}
\end{align}
i.e. $g_{\alpha\beta}(x) - A_{\alpha\beta}(x)$ may be also considered as a Fourier transform of autocorrelation function multiplied by appropriate indicator. It remains to show that this is a positive-definite function. We say that $f(x)$ is \emph{positive-definite function} if for arbitrary sequence $\{x_{k}\} \subset \domain{f}$, $k = 1, \, 2, \, ... \, , \, n$ the corresponding matrix $[f(x_{k}-x_{l})]_{k,l=1}^{n}$ is positive semi-definite\cite{AttalJoyePillet2006}, i.e. $\sum_{k,l=1}^{n} f(x_{k}-x_{l})\overline{z_k} z_{l} \geqslant 0$ for all non-zero $z \in \complex^{n}$ and for all $n \in \nat_{+}$.  Let $f(x) = r_{\alpha\beta}(x)\left(1-\chi_{[-t_0,t_0]}(x)\right)$, where $r_{\alpha\beta}(x) = \tr{\dmr \ipr{\alpha}{t}\ipr{\beta}{t-x}}$. It is known that autocorrelation functions are positive-definite\cite{AlickiLendi,BreuerPetruccione2002}, $[r_{\alpha\beta}(x_{k}-x_{l})]_{kl} \geqslant 0$ for all $\alpha,\,\beta$, so if such a sequence $\{x_{k}\}$ is chosen that $x_{k}-x_{l} \notin [-t_0, t_0]$, we have $f(x_{k}-x_{l}) = r_{\alpha\beta}(x_{k}-x_{l})$ and $[f(x_{k}-x_{l})]_{kl}$ is positive semi-definite. In other cases (i.e. some of $x_{k} - x_{l} \in [-t_0, t_0]$), situation is similar. Let $\delta x _{kl} = x_{k}-x_{l}$ and denote by $\mathcal{D}$ a set of all $\delta x_{kl}$. Define a subset $\mathcal{I} \subseteq \mathcal{D}$, $\mathcal{I} = \mathcal{D} \cap [-t_0, t_0]$. Then,
\begin{align}
	\sum_{k,l=1}^{n} f(x_{k}-x_{l})\overline{z_k} z_{l} &= \sum_{\{\delta x_{kl} \in \mathcal{I}\}} f(\delta x_{kl})\overline{z_k} z_{l} + \sum_{\{\delta x_{kl} \in \mathcal{D}\setminus\mathcal{I}\}} f(\delta x_{kl})\overline{z_k} z_{l} \\
	&= \sum_{\{\delta x_{kl} \in \mathcal{D}\setminus\mathcal{I}\}} r_{\alpha\beta}(\delta x_{kl})\overline{z_k} z_{l} \geqslant 0 \nonumber
\end{align}
since $f(\delta x_{kl}) = 0$ for $\delta x_{kl} \in \mathcal{I}$. Thus, $[f(x_{k}-x_{l})]$ is positive semi-definite for every sequence $\{x_{k}\}$ and $f(x) = \tr{\dmr \ipr{\alpha}{t}\ipr{\beta}{t-x}}\left(1-\chi_{[-t_0,t_0]}(x)\right)$ is positive-definite. Due to the Bochner's theorem, Fourier transform of $f(x)$ must create a positive function, hence $g_{\alpha\beta}(x) - A_{\alpha\beta}(x)$ is positive and the proof is complete.
\end{proof}

Applying property \ref{property:property1} from above proposition and reordering  \eqref{eq:MMEnotordered}, one obtains after some effort,
\begin{align}
	\frac{d}{dt}\ipdm{t} = &\frac{\lambda^{2}}{2} \sum_{\alpha\beta}\sum_{\{\omega\}}\sum_{q\in\integer} g_{\beta\alpha}(\omega+q\Omega) \Big( \comm{S_{\alpha}(\omega,q)\ipdm{t}}{S_{\beta}(\omega,q)^{\adj}} + \comm{S_{\alpha}(\omega,q)}{\ipdm{t}S_{\beta}(\omega,q)^{\adj}} \Big) \\
	&+ i \lambda^{2} \sum_{\alpha\beta}\sum_{\{\omega\}}\sum_{q\in\integer}\Big( \Delta_{\beta\alpha}(\omega+q\Omega)\comm{S_{\alpha}(\omega,q)\ipdm{t}}{S_{\beta}(\omega,q)^{\adj}} - \overline{\Delta_{\alpha\beta}(\omega+q\Omega)}\comm{S_{\alpha}(\omega,q)}{\ipdm{t}S_{\beta}(\omega,q)^{\adj}} \Big) \nonumber
\end{align}
The second sum may be rewritten by applying property \ref{property:property2}, namely $\Delta_{\beta\alpha}(x) = \sigma_{\beta\alpha}(x) - \frac{1}{2i}A_{\beta\alpha}(x)$ where $[\sigma_{\beta\alpha}(x)]$ and $[A_{\beta\alpha}(x)]$ were shown to be hermitian and we obtain
\begin{align}
	\frac{d}{dt}\ipdm{t} = &\lambda^{2} \sum_{\alpha\beta}\sum_{\{\omega\}}\sum_{q\in\integer} \gamma_{\beta\alpha}(\omega+q\Omega) \left( S_{\alpha}(\omega,q)\ipdm{t}S_{\beta}(\omega,q)^{\adj} - \frac{1}{2}\acomm{S_{\beta}(\omega,q)^{\adj}S_{\alpha}(\omega,q)}{\ipdm{t}}\right) \\
	&- i \lambda^{2} \sum_{\alpha\beta}\sum_{\{\omega\}}\sum_{q\in\integer} \sigma_{\alpha\beta}(\omega+q\Omega)\comm{S_{\alpha}(\omega,q)^{\adj} S_{\beta}(\omega,q)}{\ipdm{t}} \nonumber
\end{align}
where $\gamma_{\alpha\beta}(x) = g_{\alpha\beta}(x) - A_{\alpha\beta}(x)$ is hermitian and positive, according to property \ref{property:property3}. Defining new operator
\begin{equation}
	\delta H = \lambda^{2}\sum_{\alpha\beta}\sum_{\{\omega\}}\sum_{q\in\integer} \sigma_{\alpha\beta}(\omega+q\Omega)S_{\alpha}(\omega,q)^{\adj} S_{\beta}(\omega,q) 
\end{equation}
which by inspection is self-adjoint, we cast the obtained equation into familiar, Lindblad-like form
\begin{align}
	&\frac{d}{dt}\ipdm{t} = \iplgen{\ipdm{t}} = -i \comm{\delta H}{\ipdm{t}} + \tilde{\mathcal{D}}^{\prime}(\ipdm{t}), \\
	&\tilde{\mathcal{D}}^{\prime}(\ipdm{t}) = \lambda^{2} \sum_{\alpha\beta}\sum_{\{\omega\}}\sum_{q\in\integer} \gamma_{\beta\alpha}(\omega+q\Omega) \left( S_{\alpha}(\omega,q)\ipdm{t}S_{\beta}(\omega,q)^{\adj} - \frac{1}{2}\acomm{S_{\beta}(\omega,q)^{\adj}S_{\alpha}(\omega,q)}{\ipdm{t}}\right) .
\end{align}
Operator $\delta H$ is commonly known as \emph{Lamb-shift Hamiltonian}, which expresses the influence of the environment on the evolution of system of interest. By positivity of $\gamma_{\alpha\beta}$, the map $\iplgens = -i \comm{\delta H}{\,\cdot\,} + \tilde{\mathcal{D}}^{\prime}$ is explicitly time-independent generator of quantum dynamical semigroup in interaction picture, $\ipqdms{t}{t_0} = e^{(t-t_0)\iplgens}, \, t \geqslant t_0$. By general results\cite{GKS76,Lindblad76}, such maps are CPTP.
\vskip\baselineskip
In common approach, it is usually assumed that $t_0 = 0$. In such a case most of derived formulas simplify and one gets $A_{\alpha\beta} = 0$ and $\gamma_{\alpha\beta} = g_{\alpha\beta}$ is just a Fourier-transformed autocorrelation function,
\begin{equation}
	\gamma_{\alpha\beta}(x) = \int\limits_{-\infty}^{\infty} e^{-ixt^{\prime}} \tr{\dmr \ipr{\alpha}{t}\ipr{\beta}{t-t^{\prime}}} \, dt^{\prime}.
\end{equation}


\begin{thebibliography}{26}%
\makeatletter
\providecommand \@ifxundefined [1]{%
 \@ifx{#1\undefined}
}%
\providecommand \@ifnum [1]{%
 \ifnum #1\expandafter \@firstoftwo
 \else \expandafter \@secondoftwo
 \fi
}%
\providecommand \@ifx [1]{%
 \ifx #1\expandafter \@firstoftwo
 \else \expandafter \@secondoftwo
 \fi
}%
\providecommand \natexlab [1]{#1}%
\providecommand \enquote  [1]{``#1''}%
\providecommand \bibnamefont  [1]{#1}%
\providecommand \bibfnamefont [1]{#1}%
\providecommand \citenamefont [1]{#1}%
\providecommand \href@noop [0]{\@secondoftwo}%
\providecommand \href [0]{\begingroup \@sanitize@url \@href}%
\providecommand \@href[1]{\@@startlink{#1}\@@href}%
\providecommand \@@href[1]{\endgroup#1\@@endlink}%
\providecommand \@sanitize@url [0]{\catcode `\\12\catcode `\$12\catcode
  `\&12\catcode `\#12\catcode `\^12\catcode `\_12\catcode `\%12\relax}%
\providecommand \@@startlink[1]{}%
\providecommand \@@endlink[0]{}%
\providecommand \url  [0]{\begingroup\@sanitize@url \@url }%
\providecommand \@url [1]{\endgroup\@href {#1}{\urlprefix }}%
\providecommand \urlprefix  [0]{URL }%
\providecommand \Eprint [0]{\href }%
\providecommand \doibase [0]{http://dx.doi.org/}%
\providecommand \selectlanguage [0]{\@gobble}%
\providecommand \bibinfo  [0]{\@secondoftwo}%
\providecommand \bibfield  [0]{\@secondoftwo}%
\providecommand \translation [1]{[#1]}%
\providecommand \BibitemOpen [0]{}%
\providecommand \bibitemStop [0]{}%
\providecommand \bibitemNoStop [0]{.\EOS\space}%
\providecommand \EOS [0]{\spacefactor3000\relax}%
\providecommand \BibitemShut  [1]{\csname bibitem#1\endcsname}%
\let\auto@bib@innerbib\@empty
\bibitem [{\citenamefont {Lindblad}(1976)}]{Lindblad76}%
  \BibitemOpen
  \bibfield  {author} {\bibinfo {author} {\bibfnamefont {G.}~\bibnamefont
  {Lindblad}},\ }\href@noop {} {\bibfield  {journal} {\bibinfo  {journal}
  {Commun. Math. Phys.}\ }\textbf {\bibinfo {volume} {48}},\ \bibinfo {pages}
  {119} (\bibinfo {year} {1976})}\BibitemShut {NoStop}%
\bibitem [{\citenamefont {Gorini}, \citenamefont {Kossakowski},\ and\
  \citenamefont {Sudarshan}(1976)}]{GKS76}%
  \BibitemOpen
  \bibfield  {author} {\bibinfo {author} {\bibfnamefont {V.}~\bibnamefont
  {Gorini}}, \bibinfo {author} {\bibfnamefont {A.}~\bibnamefont {Kossakowski}},
  \ and\ \bibinfo {author} {\bibfnamefont {E.~C.~G.}\ \bibnamefont
  {Sudarshan}},\ }\href@noop {} {\bibfield  {journal} {\bibinfo  {journal} {J.
  Math. Phys.}\ }\textbf {\bibinfo {volume} {17}},\ \bibinfo {pages} {821}
  (\bibinfo {year} {1976})}\BibitemShut {NoStop}%
\bibitem [{\citenamefont {Davies}(1974)}]{Davies74}%
  \BibitemOpen
  \bibfield  {author} {\bibinfo {author} {\bibfnamefont {E.~B.}\ \bibnamefont
  {Davies}},\ }\href@noop {} {\bibfield  {journal} {\bibinfo  {journal}
  {Commun. Math. Phys.}\ }\textbf {\bibinfo {volume} {39}},\ \bibinfo {pages}
  {91} (\bibinfo {year} {1974})}\BibitemShut {NoStop}%
\bibitem [{\citenamefont {Davies}(1976)}]{Davies1976}%
  \BibitemOpen
  \bibfield  {author} {\bibinfo {author} {\bibfnamefont {E.~B.}\ \bibnamefont
  {Davies}},\ }\href@noop {} {\emph {\bibinfo {title} {Quantum Theory of Open
  Systems}}}\ (\bibinfo  {publisher} {Academic Press},\ \bibinfo {address}
  {London},\ \bibinfo {year} {1976})\BibitemShut {NoStop}%
\bibitem [{\citenamefont {Alicki}, \citenamefont {Lidar},\ and\ \citenamefont
  {Zanardi}(2006)}]{AlickiLidarZanardi2006}%
  \BibitemOpen
  \bibfield  {author} {\bibinfo {author} {\bibfnamefont {R.}~\bibnamefont
  {Alicki}}, \bibinfo {author} {\bibfnamefont {D.~A.}\ \bibnamefont {Lidar}}, \
  and\ \bibinfo {author} {\bibfnamefont {P.}~\bibnamefont {Zanardi}},\
  }\href@noop {} {\bibfield  {journal} {\bibinfo  {journal} {Phys. Rev. A}\
  }\textbf {\bibinfo {volume} {73}} (\bibinfo {year} {2006})}\BibitemShut
  {NoStop}%
\bibitem [{\citenamefont {Kre{\v{\i}}n}(1972)}]{Krein72}%
  \BibitemOpen
  \bibfield  {author} {\bibinfo {author} {\bibfnamefont {S.~G.}\ \bibnamefont
  {Kre{\v{\i}}n}},\ }\href@noop {} {\emph {\bibinfo {title} {Liner differential
  equations in Banach spaces}}}\ (\bibinfo  {publisher} {American Mathematical
  Society},\ \bibinfo {year} {1972})\ \bibinfo {note} {translated from
  Russian}\BibitemShut {NoStop}%
\bibitem [{\citenamefont {Chicone}(1999)}]{Chicone99}%
  \BibitemOpen
  \bibfield  {author} {\bibinfo {author} {\bibfnamefont {C.}~\bibnamefont
  {Chicone}},\ }\href@noop {} {\emph {\bibinfo {title} {Ordinary Differential
  Equations with Applications}}}\ (\bibinfo  {publisher} {Springer},\ \bibinfo
  {address} {New York},\ \bibinfo {year} {1999})\BibitemShut {NoStop}%
\bibitem [{\citenamefont {Tanabe}(1959)}]{Tanabe59}%
  \BibitemOpen
  \bibfield  {author} {\bibinfo {author} {\bibfnamefont {H.}~\bibnamefont
  {Tanabe}},\ }\href@noop {} {\bibfield  {journal} {\bibinfo  {journal} {Osaka
  Math. J.}\ }\textbf {\bibinfo {volume} {11}},\ \bibinfo {pages} {121}
  (\bibinfo {year} {1959})}\BibitemShut {NoStop}%
\bibitem [{\citenamefont {Tanabe}(1960)}]{Tanabe60}%
  \BibitemOpen
  \bibfield  {author} {\bibinfo {author} {\bibfnamefont {H.}~\bibnamefont
  {Tanabe}},\ }\href@noop {} {\bibfield  {journal} {\bibinfo  {journal} {Osaka
  Math. J.}\ }\textbf {\bibinfo {volume} {12}},\ \bibinfo {pages} {363}
  (\bibinfo {year} {1960})}\BibitemShut {NoStop}%
\bibitem [{\citenamefont {Rivas}\ and\ \citenamefont
  {Huelga}(2012)}]{HuelgaRivas2012}%
  \BibitemOpen
  \bibfield  {author} {\bibinfo {author} {\bibfnamefont {{\'{A}}.}~\bibnamefont
  {Rivas}}\ and\ \bibinfo {author} {\bibfnamefont {S.~F.}\ \bibnamefont
  {Huelga}},\ }\href@noop {} {\emph {\bibinfo {title} {Open Quantum Systems: An
  Introduction}}}\ (\bibinfo  {publisher} {Springer},\ \bibinfo {year}
  {2012})\BibitemShut {NoStop}%
\bibitem [{\citenamefont {Floquet}(1883)}]{Floquet1883}%
  \BibitemOpen
  \bibfield  {author} {\bibinfo {author} {\bibfnamefont {M.~G.}\ \bibnamefont
  {Floquet}},\ }\href@noop {} {\bibfield  {journal} {\bibinfo  {journal} {Ann.
  Ec. Norm. Suppl.}\ }\textbf {\bibinfo {volume} {12}},\ \bibinfo {pages} {47}
  (\bibinfo {year} {1883})}\BibitemShut {NoStop}%
\bibitem [{\citenamefont {Massera}\ and\ \citenamefont
  {Sch{\"{a}}ffer}(1959)}]{MasseraSchaffer59}%
  \BibitemOpen
  \bibfield  {author} {\bibinfo {author} {\bibfnamefont {J.~L.}\ \bibnamefont
  {Massera}}\ and\ \bibinfo {author} {\bibfnamefont {J.~J.}\ \bibnamefont
  {Sch{\"{a}}ffer}},\ }\href@noop {} {\bibfield  {journal} {\bibinfo  {journal}
  {Ann. Math.}\ }\textbf {\bibinfo {volume} {69}},\ \bibinfo {pages} {88}
  (\bibinfo {year} {1959})}\BibitemShut {NoStop}%
\bibitem [{\citenamefont {Sch{\"{a}}ffer}(1964)}]{Schaffer1964}%
  \BibitemOpen
  \bibfield  {author} {\bibinfo {author} {\bibfnamefont {J.~J.}\ \bibnamefont
  {Sch{\"{a}}ffer}},\ }\href@noop {} {\bibfield  {journal} {\bibinfo  {journal}
  {Bull. Amer. Math. Soc.}\ }\textbf {\bibinfo {volume} {70}},\ \bibinfo
  {pages} {243} (\bibinfo {year} {1964})}\BibitemShut {NoStop}%
\bibitem [{\citenamefont {Alicki}\ and\ \citenamefont
  {Lendi}(2006)}]{AlickiLendi}%
  \BibitemOpen
  \bibfield  {author} {\bibinfo {author} {\bibfnamefont {R.}~\bibnamefont
  {Alicki}}\ and\ \bibinfo {author} {\bibfnamefont {K.}~\bibnamefont {Lendi}},\
  }\href@noop {} {\emph {\bibinfo {title} {Quantum Dynamical Semigroups and
  Applications}}}\ (\bibinfo  {publisher} {Springer},\ \bibinfo {year}
  {2006})\BibitemShut {NoStop}%
\bibitem [{\citenamefont {Breuer}\ and\ \citenamefont
  {Petruccione}(2002)}]{BreuerPetruccione2002}%
  \BibitemOpen
  \bibfield  {author} {\bibinfo {author} {\bibfnamefont {H.-P.}\ \bibnamefont
  {Breuer}}\ and\ \bibinfo {author} {\bibfnamefont {F.}~\bibnamefont
  {Petruccione}},\ }\href@noop {} {\emph {\bibinfo {title} {The Theory of Open
  Quantum Systems}}}\ (\bibinfo  {publisher} {Oxford University Press},\
  \bibinfo {address} {New York},\ \bibinfo {year} {2002})\BibitemShut {NoStop}%
\bibitem [{\citenamefont {Attal}, \citenamefont {Joye},\ and\ \citenamefont
  {Pillet}(2006)}]{AttalJoyePillet2006}%
  \BibitemOpen
  \bibfield  {author} {\bibinfo {author} {\bibfnamefont {S.}~\bibnamefont
  {Attal}}, \bibinfo {author} {\bibfnamefont {A.}~\bibnamefont {Joye}}, \ and\
  \bibinfo {author} {\bibfnamefont {C.-A.}\ \bibnamefont {Pillet}},\
  }\href@noop {} {\emph {\bibinfo {title} {Open Quantum Systems II: The
  Markovian Approach}}}\ (\bibinfo  {publisher} {Springer},\ \bibinfo {year}
  {2006})\BibitemShut {NoStop}%
\bibitem [{\citenamefont {Alicki}(1979)}]{Alicki79}%
  \BibitemOpen
  \bibfield  {author} {\bibinfo {author} {\bibfnamefont {R.}~\bibnamefont
  {Alicki}},\ }\href@noop {} {\bibfield  {journal} {\bibinfo  {journal} {J.
  Phys. A: Math. Gen.}\ }\textbf {\bibinfo {volume} {12}},\ \bibinfo {pages}
  {L103} (\bibinfo {year} {1979})}\BibitemShut {NoStop}%
\bibitem [{\citenamefont {D{\"{u}}mke}\ and\ \citenamefont
  {Spohn}(1979)}]{DumkeSpohn1979}%
  \BibitemOpen
  \bibfield  {author} {\bibinfo {author} {\bibfnamefont {R.}~\bibnamefont
  {D{\"{u}}mke}}\ and\ \bibinfo {author} {\bibfnamefont {H.}~\bibnamefont
  {Spohn}},\ }\href@noop {} {\bibfield  {journal} {\bibinfo  {journal} {Z.
  Phys. B}\ }\textbf {\bibinfo {volume} {34}},\ \bibinfo {pages} {419}
  (\bibinfo {year} {1979})}\BibitemShut {NoStop}%
\bibitem [{\citenamefont {Szczygielski}, \citenamefont {Gelbwaser-Klimovsky},\
  and\ \citenamefont {Alicki}(2013)}]{Szczygielski2013}%
  \BibitemOpen
  \bibfield  {author} {\bibinfo {author} {\bibfnamefont {K.}~\bibnamefont
  {Szczygielski}}, \bibinfo {author} {\bibfnamefont {D.}~\bibnamefont
  {Gelbwaser-Klimovsky}}, \ and\ \bibinfo {author} {\bibfnamefont
  {R.}~\bibnamefont {Alicki}},\ }\href@noop {} {\bibfield  {journal} {\bibinfo
  {journal} {Phys. Rev. E}\ }\textbf {\bibinfo {volume} {87}} (\bibinfo {year}
  {2013})}\BibitemShut {NoStop}%
\bibitem [{\citenamefont {Gelbwaser-Klimovsky}, \citenamefont {Alicki},\ and\
  \citenamefont {Kurizki}(2013)}]{AlickiGelbwaserKurizki2013}%
  \BibitemOpen
  \bibfield  {author} {\bibinfo {author} {\bibfnamefont {D.}~\bibnamefont
  {Gelbwaser-Klimovsky}}, \bibinfo {author} {\bibfnamefont {R.}~\bibnamefont
  {Alicki}}, \ and\ \bibinfo {author} {\bibfnamefont {G.}~\bibnamefont
  {Kurizki}},\ }\href@noop {} {\bibfield  {journal} {\bibinfo  {journal} {Phys.
  Rev. E}\ }\textbf {\bibinfo {volume} {87}},\ \bibinfo {pages} {012140}
  (\bibinfo {year} {2013})}\BibitemShut {NoStop}%
\bibitem [{\citenamefont {Gelbwaser-Klimovsky}\ \emph
  {et~al.}(2013)\citenamefont {Gelbwaser-Klimovsky}, \citenamefont {Erez},
  \citenamefont {Alicki},\ and\ \citenamefont
  {Kurizki}}]{GelbwaserErezAlickiKurizki2013}%
  \BibitemOpen
  \bibfield  {author} {\bibinfo {author} {\bibfnamefont {D.}~\bibnamefont
  {Gelbwaser-Klimovsky}}, \bibinfo {author} {\bibfnamefont {N.}~\bibnamefont
  {Erez}}, \bibinfo {author} {\bibfnamefont {R.}~\bibnamefont {Alicki}}, \ and\
  \bibinfo {author} {\bibfnamefont {G.}~\bibnamefont {Kurizki}},\ }\href@noop
  {} {\bibfield  {journal} {\bibinfo  {journal} {Phys. Rev. A}\ }\textbf
  {\bibinfo {volume} {88}},\ \bibinfo {pages} {022112} (\bibinfo {year}
  {2013})}\BibitemShut {NoStop}%
\bibitem [{\citenamefont {Alicki}, \citenamefont {Levy},\ and\ \citenamefont
  {Kosloff}(2012)}]{LevyAlickiKosloff2006}%
  \BibitemOpen
  \bibfield  {author} {\bibinfo {author} {\bibfnamefont {R.}~\bibnamefont
  {Alicki}}, \bibinfo {author} {\bibfnamefont {A.}~\bibnamefont {Levy}}, \ and\
  \bibinfo {author} {\bibfnamefont {R.}~\bibnamefont {Kosloff}},\ }\href@noop
  {} {\bibfield  {journal} {\bibinfo  {journal} {Phys. Rev. E}\ }\textbf
  {\bibinfo {volume} {85}},\ \bibinfo {pages} {061126} (\bibinfo {year}
  {2012})}\BibitemShut {NoStop}%
\bibitem [{\citenamefont {Kol{\'a\v{r}}}\ \emph {et~al.}(2012)\citenamefont
  {Kol{\'a\v{r}}}, \citenamefont {Gelbwaser-Klimovsky}, \citenamefont
  {Alicki},\ and\ \citenamefont {Kurizki}}]{AlickiKolarGelbwaserKurizki2012}%
  \BibitemOpen
  \bibfield  {author} {\bibinfo {author} {\bibfnamefont {M.}~\bibnamefont
  {Kol{\'a\v{r}}}}, \bibinfo {author} {\bibfnamefont {D.}~\bibnamefont
  {Gelbwaser-Klimovsky}}, \bibinfo {author} {\bibfnamefont {R.}~\bibnamefont
  {Alicki}}, \ and\ \bibinfo {author} {\bibfnamefont {G.}~\bibnamefont
  {Kurizki}},\ }\href@noop {} {\bibfield  {journal} {\bibinfo  {journal} {Phys.
  Rev. Lett.}\ }\textbf {\bibinfo {volume} {109}},\ \bibinfo {pages} {090601}
  (\bibinfo {year} {2012})}\BibitemShut {NoStop}%
\bibitem [{\citenamefont {Frigerio}(1978)}]{Frigerio1978}%
  \BibitemOpen
  \bibfield  {author} {\bibinfo {author} {\bibfnamefont {A.}~\bibnamefont
  {Frigerio}},\ }\href@noop {} {\bibfield  {journal} {\bibinfo  {journal}
  {Comm. Math. Phys.}\ }\textbf {\bibinfo {volume} {63}},\ \bibinfo {pages}
  {269} (\bibinfo {year} {1978})}\BibitemShut {NoStop}%
\bibitem [{\citenamefont {Chu}\ and\ \citenamefont
  {Telnov}(2004)}]{ChuTelnov2004}%
  \BibitemOpen
  \bibfield  {author} {\bibinfo {author} {\bibfnamefont {S.-I.}\ \bibnamefont
  {Chu}}\ and\ \bibinfo {author} {\bibfnamefont {D.~A.}\ \bibnamefont
  {Telnov}},\ }\href@noop {} {\bibfield  {journal} {\bibinfo  {journal} {Phys.
  Rep.}\ }\textbf {\bibinfo {volume} {390}},\ \bibinfo {pages} {1} (\bibinfo
  {year} {2004})}\BibitemShut {NoStop}%
\bibitem [{\citenamefont {Ernst}, \citenamefont {Samoson},\ and\ \citenamefont
  {Meier}(2005)}]{ErnstSamosonMeier2005}%
  \BibitemOpen
  \bibfield  {author} {\bibinfo {author} {\bibfnamefont {M.}~\bibnamefont
  {Ernst}}, \bibinfo {author} {\bibfnamefont {A.}~\bibnamefont {Samoson}}, \
  and\ \bibinfo {author} {\bibfnamefont {B.~H.}\ \bibnamefont {Meier}},\
  }\href@noop {} {\bibfield  {journal} {\bibinfo  {journal} {J. Chem. Phys.}\
  }\textbf {\bibinfo {volume} {123}},\ \bibinfo {pages} {064102} (\bibinfo
  {year} {2005})}\BibitemShut {NoStop}%
\end{thebibliography}

\providecommand{\noopsort}[1]{}\providecommand{\singleletter}[1]{#1}%

\end{document}